\documentclass[11pt,letterpaper]{article}

\usepackage[utf8]{inputenc}
\usepackage[ruled,vlined,commentsnumbered,titlenotnumbered]{algorithm2e}
\usepackage[subtle, mathspacing=normal]{savetrees} 
\usepackage{url}
\usepackage{fullpage}
\usepackage{enumitem}
\usepackage{pslatex} \usepackage{mathdots} \usepackage{amsmath}
\usepackage{algorithmic}
\usepackage{amssymb}
\usepackage{amsthm}
\usepackage{color}
\usepackage{array}
\usepackage{xy}
\usepackage{setspace}
\usepackage{varwidth}\usepackage{multicol}

\usepackage[normalem]{ulem} 

\usepackage{hyperref}
\usepackage{cite}
\usepackage[margin=1in,letterpaper]{geometry}
\usepackage{amsmath}

\usepackage{accents}

\newcommand{\defn}[1]{\emph{\textbf{{#1}}}}

\newcommand{\E}{\mathbb{E}}

\renewcommand{\paragraph}[1]{\vspace{.3 cm} \noindent \textbf{\boldmath #1} }

\usepackage{thmtools}
\usepackage{thm-restate}

\newtheoremstyle{slanted}{3pt}{3pt}{\slshape}{}{\bfseries}{.}{.5em}{}\theoremstyle{slanted}
\newtheorem{theorem}{Theorem}\newtheorem{lemma}[theorem]{Lemma}
\newtheorem{observation}[theorem]{Observation}

\newtheorem{proposition}[theorem]{Proposition}
\newtheorem{corollary}[theorem]{Corollary}

\newtheorem{remark}[theorem]{Remark}
\newtheorem{property}[theorem]{Property}

\DeclareMathOperator{\loglog}{loglog}

\title{Online List Labeling: Breaking the $\log^2 n$ Barrier}

\author{Michael A.~Bender\\Stony Brook University \and Alex Conway\\VMWare Research \and Mart\'{\i}n Farach-Colton\\Rutgers University \and Hanna Koml\'os\\Rutgers University \and William Kuszmaul\\MIT \and Nicole Wein\\DIMACS}

\begin{document}
\date{}
\maketitle

\begin{abstract}

The online list-labeling problem is an algorithmic primitive with a large literature of upper bounds, lower bounds, and applications. The goal is to store a dynamically-changing set of $n$ items in an array of $m$ slots, while maintaining the invariant that the items appear in sorted order, and while minimizing the \emph{relabeling cost}, defined to be the number of items that are moved per insertion/deletion.

For the linear regime, where $m = (1 + \Theta(1)) n$, an upper bound of $O(\log^2 n)$ on the relabeling cost has been known since 1981. A lower bound of $\Omega(\log^2 n)$ is known for deterministic algorithms and for so-called \emph{smooth} algorithms, but the best general lower bound remains $\Omega(\log n)$. The central open question in the field is whether $O(\log^2 n)$ is optimal for all algorithms. 
    
In this paper, we give a randomized data structure that achieves an expected relabeling cost of $O(\log^{3/2} n)$ per operation. More generally, if $m = (1 + \epsilon) n$ for $\epsilon = O(1)$, the expected relabeling cost becomes $O(\epsilon^{-1} \log^{3/2} n)$.  

Our solution is \emph{history independent}, meaning that the state of the data structure is independent of the order in which items are inserted/deleted. For history-independent data structures, we also prove a matching lower bound: for all $\epsilon$ between $1 / n^{1/3}$ and some sufficiently small positive constant, the optimal expected cost for history-independent list-labeling solutions is $\Theta(\epsilon^{-1}\log^{3/2} n)$.

\end{abstract}

 \thispagestyle{empty}
\newpage
\pagenumbering{arabic}

\section{Introduction}\label{sec:intro}

The online \defn{list-labeling problem} is one of the most basic and well-studied algorithmic primitives in data structures, with an extensive literature spanning upper bounds~\cite{ItaiKoRo81,Willard82, Willard86, Willard92,BenderCoDe02twosimplified, BenderFiGi17,ItaiKa07,AnderssonLa90,GalperinR93,BabkaBCKS19, BenderHu07, Katriel02,BenderFiGi05, BenderDeFa05}, lower bounds~\cite{dietz1990lower,dietz1994tight,dietz2004tight,zhang1993density,Saks18,BulanekKoSa12}, variants~\cite{BenderHu07,BenderBeJo16,DevannyFiGo17,Dietz82,Andersson89,AnderssonLa90,GalperinR93,Raman99}, and open-problem surveys~\cite{Saks18, gal2021computational}. The problem has been independently re-discovered in many different contexts~\cite{Willard81, Andersson89, GalperinR93, Raman99}, and it has found extensive applications to areas such as ordered maintenance~\cite{Dietz82, BenderCoDe02twosimplified, BenderFiGi17, BenderCoDe02a}, cache-oblivious data structures~\cite{BenderDeFa05, BrodalFaJa02, BenderDuIa04, BenderFiGi05, BenderFaKu06}, dense file maintenance~\cite{Willard81,Willard86,Willard92,Willard82}, applied graph algorithms~\cite{WheatmanX21, WheatmanX18, WheatmanB21, PandeyWXB21, LeoB21, LeoB19fastconcurrent}, etc.  (For a detailed discussion of related work and applications, see Section~\ref{sec:related}.)

The list-labeling problem was originally formulated~\cite{ItaiKoRo81} as follows. An algorithm must store a set of $n$ elements (where $n$ changes over time) in sorted order in an array of $m \ge n$ slots. Elements are inserted and deleted over time, with each insertion specifying the new element's \defn{rank} $r \in \{1, 2, \ldots, n + 1\}$ among the other elements that are present (e.g., inserting at rank $1$ means that the inserted element is the new smallest element). To keep the elements in sorted order in the array, the algorithm must sometimes move elements around. The \defn{cost} of an algorithm is the number of elements moved during the insertions/deletions.\footnote{To accommodate the many ways in which list labeling is used, some works describe the problem in a more abstract (but equivalent) way: the list-labeling algorithm must dynamically assign each element $x$ a label $\ell(x) \in \{1, 2, \ldots, m\}$ such that $x \prec y \iff \ell(x) < \ell(y)$, and the goal is to minimize the number of elements that are \emph{relabeled} per insertion/deletion---hence the name of the problem.}

The list-labeling problem is well understood in the regime where $m \gg n$. In the \defn{pseudo-exponential regime}, when $\frac{m}{n} = 2^{n^{\Omega(1)}}$, it is possible to achieve $O(1)$ amortized cost per operation~\cite{BabkaBCKS19}.  In the \defn{polynomial regime}, when $\frac{m}{n} = n^{\Theta(1)}$, the amortized cost becomes $O(\log n)$ \cite{Kopelowitz12,Andersson89,GalperinR93}. These bounds are known to be tight for both deterministic and randomized algorithms~\cite{BulanekKoSa13, BabkaBCKS12, BabkaBCKS19}.

It has remained an open problem, however, what happens in the \defn{linear regime}, where $m = (1 + \epsilon) n$ for some $\epsilon = \Theta(1)$.  In 1981, Itai, Konheim, and Rodeh~\cite{ItaiKoRo81} showed how to achieve amortized cost $O(\log^2 n)$, and posed as an open question whether any algorithm could do better. Despite a great deal of subsequent work on alternative solutions (including deterministic, randomized, and deamortized algorithms) for the same problem~\cite{Willard82, Willard86, Willard92, BenderFiGi17, ItaiKa07, BenderHu07, BenderBeJo16, BenderCoDe02twosimplified, Katriel02}, the bound of $O(\log^2 n)$ has remained unimproved for four decades.

Starting in 1990, there has been a long line of work towards establishing a matching $\Omega(\log^2 n)$ lower bound~\cite{dietz1990lower, dietz1994tight, dietz2004tight, BulanekKoSa12, BulanekKoSa13}. It is known that any deterministic algorithm requires $\Omega(\log^2 n)$ amortized cost per insertion~\cite{BulanekKoSa12}. And the same lower bound holds for \defn{smooth} algorithms, where the relabelings are restricted to evenly rebalance elements across a contiguous subarray~\cite{dietz1990lower}. This second lower bound is surprisingly strong: it applies even to randomized algorithms and even to the offline problem, where the entire sequence of operations is known \textsl{a priori}.  However, the best general lower bound remains $\Omega(\log n)$~\cite{BulanekKoSa13}. 

These lower bounds tell us that, if an algorithm is to beat the $O(\log^2 n)$ bound, then the algorithm must be both randomized and non-smooth. Whether or not any such algorithm is possible has remained the central open question~\cite{ItaiKoRo81, dietz1990lower, dietz1994tight, dietz2004tight, BulanekKoSa12, BulanekKoSa13} in this research area (see also discussion of the problem in open-problem surveys and textbooks~\cite{Saks18, gal2021computational, mehlhorn2008algorithms}).  Several sources~\cite{dietz1990lower, dietz1994tight, dietz2004tight} have conjectured that $\Theta(\log^2 n)$ cost is optimal in general.

\paragraph{Breaking through the $\log^2 n$ barrier.} 
We present a randomized list-labeling algorithm that achieves expected cost $O(\log^{3/2} n)$ per insertion/deletion in the linear regime (Corollary \ref{thm:linearcase}). 
In breaking through the $\log^2 n$ barrier, we establish that there is a fundamental gap between deterministic and randomized algorithms for online list labeling. 
Our result is the first asymptotic improvement in the linear regime in the 40-year history of the problem. 

The original $O(\log^2 n)$ upper bound by Itai et al.~\cite{ItaiKoRo81} also extends to the \defn{dense regime} of $\epsilon = o(1)$, where the bound on amortized cost becomes $O(\epsilon^{-1} \log^2 n)$~\cite{AnderssonLa90, BirdSa07, zhang1993density}. Extending our algorithm to the same regime, we achieve expected cost $O(\epsilon^{-1} \log^{3/2} n)$ (Theorem \ref{thm:shiftedzeno}).

Applying our result to the insertion-only setting, the array can be filled from empty to full (i.e., $n = m$) in total expected time $O(n \log^{2.5} n)$ (Corollary \ref{cor:fill_array}). This improves over the previous state of the art of $O(n \log^{3} n)$, which was known to be optimal for deterministic algorithms~\cite{BulanekKoSa12}---again we have a separation between what can be achieved with deterministic and randomized algorithms. 

A surprising aspect of our results is how they contrast with the polynomial regime $m = n^{1 + \Theta(1)}$, where randomized and deterministic algorithms are asymptotically equivalent \cite{BulanekKoSa13, BabkaBCKS12, BabkaBCKS19}. Our final upper-bound result considers a continuum between these regimes, where $m = \omega(n) \cap n^{o(1)}$. In this \defn{sparse regime} there is a folklore bound~\cite{Kopelowitz12,Andersson89,GalperinR93} of $O\!\left(\frac{\log^2 n}{\log (m / n)}\right)$, which continuously deforms between $O(\log^2 n)$ for the linear regime and $O(\log n)$ for the polynomial regime. Using our techniques (Theorem \ref{thm:sparse}), we achieve expected cost $$O\!\left(\frac{\log^{3/2} n}{\sqrt{\log (m / n)}}\right).$$ 
Thus we achieve asymptotic improvements for list labeling for all $m = n^{1 + o(1)}$. 

\paragraph{An unexpected tool: history independence.} One research area that our algorithms build directly upon is the study of history-independent data structures: a data structure is said to be \defn{history independent}~\cite{Micciancio97, NaorTe01} if its current state reveals nothing about the history of the past operations beyond the current set of elements that are present.

History independence is typically viewed as a security guarantee, with the intent being to minimize the risk incurred by a security breach. Research on history-independent data structures~\cite{Micciancio97, NaorTe01, HartlineHoMo05, BuchbinderPe03, BlellochGo07, NaorSeWi08, Golovin09, Golovin10} (as well as on history-independent list labeling~\cite{BenderBeJo16} specifically) has focused on history independence as an \emph{end goal}, with the question being whether history independence can be achieved without any increase in running time.

We find that, in the context of list labeling, history independence is actually a valuable \emph{algorithmic tool} for building faster randomized data structures. History independence allows for us to have a data structure with vulnerabilities (i.e., certain spots where an insertion would be expensive) while (1) keeping those vulnerabilities hidden from the adversary; and (2) preventing the adversary from having any control over where those vulnerabilities appear. This simple paradigm plays an important role in allowing our randomized data structures to bypass the $\log^2 n$ barrier.

\paragraph{A matching lower bound for history-independent data structures.}
Finally, we show that our bounds in the dense regime are asymptotically optimal for any history-independent data structure: there exists a positive constant $c$ such that, for all $1 / n^{1/3} \le \epsilon \le 1/c$, the expected insertion/deletion cost when $m = (1 + \epsilon) n$ is necessarily at least $\Omega(\epsilon^{-1} \log^{3/2} n)$ for any history-independent data structure (Theorem \ref{thm:lower}).

This means that, if there exists a randomized data structure that achieves better bounds than those in this paper, then the data structure must fundamentally be \emph{adaptive} in how it responds to the history of the operations being performed. Of course, by being adaptive, such a data structure would also implicitly surrender the structural anonymity that history independence offers, revealing information about where the ``hotspots'' are within the data structure. Our results suggest that $\log^{3/2} n$ is a potentially fundamental barrier---whether or not the bounds achieved in this paper are optimal in general remains an enticing open problem.

\paragraph{Paper outline.}
The rest of this paper proceeds as follows. Section \ref{sec:prelims} gives preliminaries. Section \ref{sec:intuition} gives an intuitive overview of our upper bound and proof techniques. Sections \ref{sec:zeno} and \ref{sec:dense} present our upper bounds for the linear/dense regime. A key technical idea is to control the local density of the array via a random process that we call a Zeno random walk---we describe and analyze this random walk in Section \ref{sec:zeno}. Section~\ref{sec:dense} then gives our (history-independent) list-labeling data structure and uses the bounds on Zeno random walks to analyze it.  Section~\ref{sec:lower} presents our lower bound for history-independent list-labeling data structures.  Section~\ref{sec:sparse} gives a black-box reduction for transforming dense list-labeling solutions into sparse list-labeling solutions---this yields our upper bound for the sparse regime.  Finally, Section~\ref{sec:related} discusses related work in more detail.

\section{Preliminaries}\label{sec:prelims}

In this section, we formally define the list-labeling problem and history independence---we then outline the classical $O(\log^2 n)$ solution~\cite{ItaiKoRo81} and a more recent history-independent variation on that solution~\cite{BenderBeJo16}.

\paragraph{The list-labeling problem.}
A list-labeling data structure stores a dynamically changing set of size $n \le m$ in an array of $m$ slots. It supports two operations: 
\begin{itemize}
\item \textsc{Insert($r$), $r \in \{1, 2, \ldots, n + 1\}$:} This operation adds an element whose rank is $r$. This increments $n$ and also increments the ranks of each of the elements whose ranks were formerly in $\{r , \ldots, n\}$. 
\item \textsc{Delete($r$), $r \in \{1, 2, \ldots, n \}$:} This operation removes the element whose rank is $r$. This decrements $n$ and also decrements the ranks of each of the elements whose ranks were formally in $\{r + 1, \ldots, n\}$. 
\end{itemize}
The list-labeling algorithm must maintain the invariant that the elements appear in sorted order (by rank) within the array. The \defn{cost} of an insertion/deletion is the number of elements that are moved within the array during the insertion/deletion (including the element being inserted/deleted).
In the case where $n = \Omega(m)$, we will further guarantee (for our upper bounds) that the maximum gap between any two consecutive elements in the array is at most $ O (1) $ positions---this extra guarantee is often required for applications of list labeling in which algorithms perform range queries within the array, e.g., \cite{WheatmanX18, PandeyWXB21, WheatmanX21, BenderDuIa04, BenderDeFa05}. 

We will typically use an additional parameter $\epsilon$ such that either $n \le (1 - \epsilon) m$ or $m \ge (1 + \epsilon) n$ (the specific convention that we follow will differ from section to section to optimize for simplifying the algebraic manipulation in each section). 

From the perspective of the list-labeling data structure, the elements that it stores are black boxes---the only information that the data structure knows about its elements is their sorted order. This allows for list labeling to be used in applications where the elements are from arbitrary universes.

Finally, it is important to emphasize that the insertions/deletions are performed by an \defn{oblivious adversary}, who does not get to see the random decisions made by the list-labeling data structure. If the adversary were to be adaptive, then, trivially, no randomized list-labeling data structure could incur expected cost any better than the worst-case cost of the best deterministic list-labeling data structure.

\paragraph{History independence.}
A data structure is said to be  \defn{history independent} \cite{Micciancio97, NaorTe01, HartlineHoMo05, BuchbinderPe03, BlellochGo07, NaorSeWi08, Golovin09, Golovin10, BenderBeJo16} if, given access to the current state of the data structure, the only information that an adversary can deduce is the current set of elements; that is, the adversary gains no information about the history of operations performed. In the list-labeling data structure the current set of elements is specified only by their relative ranks, so the only information that an adversary can deduce is the \emph{number} of elements.

History independence plays an important supporting role throughout this paper. Indeed, although history independence does not on its own improve the asymptotics of list labeling, it does create a natural abstraction for how to separate the behavior of the data structure that we are designing from the actions of the user.

There are several basic mathematical properties of history independence that will be useful in both our upper and lower bounds. Define the \defn{array configuration} of a list-labeling data structure to be the boolean vector in $\{0, 1\}^m$ indicating which $n$ positions of the array contain elements. 
We have the following properties of a history-independent data structure for list-labeling:

\begin{property}\phantom{f}
\begin{enumerate}[noitemsep,label=(\alph*)]
\item Whenever the array contains $n$ elements, its array configuration $A$ satisfies $A \sim \mathcal{C}_{n, m}$, where $\mathcal{C}_{n, m}$ is some probability distribution over array configurations.
\item Whenever an insertion is performed at rank $r \in \{1, 2, \ldots, n + 1\}$ in an array with $n$ elements, the array configurations $A_0$ and $A_1$ before and after the insertion satisfy $(A_0, A_1) \sim \mathcal{I}_{n, m, r}$, where $\mathcal{I}_{n, m, r}$ is a joint distribution between $\mathcal{C}_{n, m}$ and $\mathcal{C}_{n + 1, m}$.\footnote{A probability distribution $\cal X$ is a joint distribution between distributions $\cal A$ and $\cal B$ if $(A, B) \sim {\cal X} \implies A \sim {\cal A}, \, B \sim {\cal B}$.}
\item Whenever a deletion is performed at rank $r \in \{1, 2, \ldots, n + 1\}$ in an array with $n + 1$ elements, the array configurations $A_1$ and $A_0$ before and after the deletion satisfy $(A_0, A_1) \sim \mathcal{D}_{n, m, r}$, where $\mathcal{D
}_{n, m, r}$ is a joint distribution between $\mathcal{C}_{n, m}$ and $\mathcal{C}_{n + 1, m}$.
\end{enumerate}
\end{property}

These properties imply that the (probability distribution on the) behavior of the algorithm on any given operation is fully determined by $n$, $m$, the operation (insertion or deletion), and  the rank $r$ of the element being inserted/deleted. In our upper bounds, we will further have that $\mathcal{D}_{n , m, r} = \mathcal{I}_{n, m, r}$; we call any list-labeling data structure with this property \defn{insertion/deletion symmetric}.

\subsection{The Classical Solution and its History-Independent Analogue}

\vspace{-.2 cm}

\paragraph{List labeling with weight-balanced trees.} The original solution to list labeling~\cite{ItaiKoRo81}, due to Itai et al.~\cite{ItaiKoRo81} in 1981, can be described in terms of weight-balanced trees~\cite{GalperinR93,NievergeltRe72,NievergeltRe73}. For brevity, we will describe the solution here for the linear regime, where $m = (1 + \Theta(1)) n$, but the same solution directly generalizes to all regimes from dense ($n = (1 - \epsilon) m$) to polynomial ($m = n^{1+\Theta(1)}$). 

Consider an array of size $m$, and impose a tree structure on it, where the root node represents the entire array, the nodes in the $i$-th level of the tree represent disjoint sub-arrays of size $m / 2^{i - 1}$, and the leaf nodes represent sub-arrays of size $\Theta(\log n)$. We keep the tree tightly weight balanced, meaning that, for any pair of sibling nodes $x$ and $y$, their densities are always within a $1 \pm O(1 / \log n)$ factor of each other. In particular, whenever an insertion or deletion breaks this invariant for some pair of siblings $x$ and $y$, we take the elements in the sub-array $x \cup y$ and rearrange them to be distributed evenly across that sub-array.\footnote{This approach is both deterministic and smooth, and thus consistent with the assumptions made by lower bounds~\cite{dietz1990lower, dietz1994tight, dietz2004tight, BabkaBCKS12}.}

This tight weight balancing ensures that \emph{all} of the nodes in the tree have densities that are within a factor of $(1 + O(1 / \log n))^{O(\log n)} = O(1)$ of each other. By selecting the constants in the algorithm appropriately, one can ensure that every leaf has more slots than it has elements, which guarantees the correctness of the data structure. On the other hand, in order to maintain such tight weight balancing, one must rebuild nodes a factor of $O(\log n)$ more often than in a standard weight-balanced binary search tree~\cite{GalperinR93,NievergeltRe72,NievergeltRe73}, leading to an amortized cost of $O(\log^2 n)$.

Intuitively, the above data structure would seem to be the asymptotically optimal approach to maintaining tightly-balanced densities within an array---the known lower bounds for list labeling~\cite{dietz1990lower, dietz1994tight, dietz2004tight, BulanekKoSa12} confirm that this is the case for both deterministic and smooth data structures. The upper bounds in this paper reveal that, perhaps surprisingly, it is not the case for randomized data structures. Randomization fundamentally reduces the cost to maintain a tightly weight-balanced tree.

\paragraph{History-independent list labeling.} 
To understand how history independence can be achieved in the context of list labeling, it is helpful to first understand it in the context of balanced binary search trees. The classic example of a balanced binary tree with a history-independent topology is the \defn{randomized binary search tree} \cite{AragonSe89, SeidelAr96} (or, similarly, the treap \cite{AragonSe89, SeidelAr96}), which maintains as an invariant that, at any given moment, the structure of the tree is random (i.e., that within each subtree, the root of that subtree is a random element). This can be achieved with reservoir sampling \cite{AragonSe89, SeidelAr96, vitter1985random, li1994reservoir, BenderBeJo16}---in particular, whenever a new item is added to a subtree of (former) size $r$, the element becomes the new root with probability $1 / (r + 1)$ (in which case the subtree is rebuilt from scratch). This simple approach yields an expected time of $O(\log n)$ per operation.

As shown by Bender et al. \cite{BenderBeJo16}, the same basic approach can be used to achieve history-independent list labeling. Now, the tree is random across all \emph{tightly balanced trees}---that is, within each subtree $T$ containing elements $x_1 < x_2 < \cdots < x_k$, the root is a random element $x_i$ of those satisfying $|i - k / 2| \in O(k / \log n)$. As before, this structure can be maintained using reservoir sampling. However, the restriction that the tree must be tightly balanced increases the frequency with which subtrees are rebuilt, so that the expected cost per operation becomes $O(\log^2 n)$, just as for the standard solution to list labeling. 

 \section{Technical Overview}\label{sec:intuition}

In this section, we present an intuitive overview of our upper bound and proof techniques. Comprehensive technical details can be found in 
Sections~\ref{sec:zeno} and~\ref{sec:dense}. For simplicity, we shall assume in this section that $m = 2n$.

Intuitively, our starting point is the history-independent list labeling solution by Bender, et al. \cite{BenderBeJo16}. As described in Section \ref{sec:prelims}, in \cite{BenderBeJo16}, the root of any subtree of size $k$ is a random element of the middle $O(k/\log n)$ elements of the subtree. We call this middle set of elements the \emph{candidate set}. 

A natural idea for decreasing the cost of this algorithm is to increase the size of the candidate set to $\delta k$ for some $\delta = \omega(1 / \log n)$. This way, the root would be resampled less often, resulting in fewer total rebalances. However, there is a problem with this approach: the subarrays representing the nodes in the $i$-th level of the tree have densities bounded between $\frac{1}{2} (1 - \delta)^i$ and $\frac{1}{2} (1 + \delta)^i$, but this means that nodes in the $\Theta(\log n)$-th level can overflow with a density of $\frac{1}{2} (1 + \delta)^{\Theta(\log n)} = \omega(1)$. Thus, having $\delta = \omega(1/ \log n)$ violates the correctness of the algorithm.

Notice, however, that \emph{most} nodes in the $i$-th level of the tree avoid a density of the form $\frac{1}{2} (1 + \delta)^i$. Indeed, if we were to perform a random walk down the tree, then the node that we encountered on our $i$-th step would likely have a density bounded above by $\frac{1}{2} (1 + \delta)^{O(\sqrt{i})}$. This means that, if we only wanted \emph{most nodes} to behave well, then we could set $\delta$ close to $\frac{1}{\sqrt{\log n}}$. 

In order to obtain the benefits of $\delta \approx 1 / \sqrt{\log n}$ while maintaining the correctness of $\delta \approx 1 / \log n$, we smoothly adjust the candidate set size for each subtree as a function of the subtree's density. We show that almost all subtrees are sparse enough to support a ``large'' candidate set ($\delta \approx 1 / \sqrt{\log n}$), while only a small fraction of subtrees require ``smaller'' candidate sets (with $\delta$ closer to $1 / \log n$). This means that most parts of the array support fast insertions/deletions, while only a small portion of the array is slow to insert/delete to.

While we have made progress by ensuring that most of the array can support fast updates, this is not sufficient to prove the final bound. Specifically, if the adversary \emph{knows} which parts of the array are slow to update, they could simply focus all of their insertions/deletions on these slow parts of the array, causing the total cost to be large. 
Instead, we would like to \emph{hide} the slow parts of the array from the adversary. 
More precisely, we are concerned about two distinct problems: the adversary could \emph{create} dense regions through their insertion sequence (e.g., by concentrating insertions in one location), or, the adversary could \emph{detect} dense regions created by the algorithm (e.g., through prior knowledge of the algorithm's distribution of states.)  

History independence comes into play in guarding against these problems. 
By definition, the first problem cannot happen with a history-independent algorithm, since the configuration of the array does not depend on the adversary's specific sequence of insertions. 
For the second problem, we add an additional layer of randomness called a \emph{random shift}. 
At the start of the algorithm, we insert random number $k \in [m]$ of dummy elements at the front of the array, and $m-k$ at the end.
This converts a potentially adversarial insertion at rank $j$ to a uniformly random insertion of rank between $j$ and $j+m$.
Together with history independence, the random shift ensures that the adversary cannot target specific regions of the array. 

To analyze our algorithm, we introduce the notion of a \emph{Zeno random walk}, which is a special type of bounded random walk where the step size decreases as the distance to a boundary decreases. The Zeno walk captures the way in which the densities of subproblems evolve if we perform a random walk down our tree. Our analysis of this random walk (Proposition~\ref{prop:zeno2}) allows us to bound the cost of a random insertion (Lemma~\ref{lem:randominsert}). 
Finally, we extend this analysis for a \emph{random} insertion to an \emph{arbitrary} insertion using the ideas outlined above of history independence and a random shift, achieving an expected $O(\log^{3/2} n)$ cost for any insertion/deletion.

\section{Zeno's Random Walk}\label{sec:zeno}
This section describes and analyzes a simple but somewhat unusual type of random walk that we will refer to as a \defn{Zeno walk}---this random walk will play an important algorithmic role in later sections.

Let $\delta \in (0, 1/2]$. A Zeno walk $Z_0, Z_1, Z_2, \ldots$ starts at $Z_0 = 0$ and deterministically satisfies $Z_i \in (-1, 1)$ for all $i$. We define $\alpha_i = 1 - |Z_i|$ to be the distance between $Z_i$ and the nearest boundary $1$ or $-1$. We determine $Z_{i + 1}$ from $Z_{i}$ as follows:
\begin{itemize}
\item An adaptive adversary selects a quantity $\delta_i \le \delta$, possibly as a function of $Z_0, Z_1, \ldots, Z_i$. 
\item $Z_{i + 1}$ is then set to be one of $Z_i + \alpha_{i} \delta_i$ or $Z_i - \alpha_i \delta_i$, each with equal probability. 
\end{itemize}

What makes the Zeno walk unusual is that, the closer it gets to $-1$ or $1$, the smaller its steps become (since the $i$-th step has its size multiplied by $\alpha_i$). The result is that (as in Zeno's paradox), the walk can get arbitrarily close to $\pm 1$ but can never reach $\pm 1$. 

We will be interested in Zeno walks $Z_1, \ldots, Z_\ell$ where the relationship between $\delta$ and the length $\ell$ of the walk is $\delta = O(1 / \sqrt{\ell})$. To gain some intuition here, consider the case where $\delta_i = \delta = 1 / \sqrt{\ell}$ for all $i$, and let us compare the Zeno walk $Z_1, \ldots, Z_\ell$ to a standard unbiased random walk $X_1, \ldots, X_{\ell}$ that changes by $\pm 1 / \sqrt{\ell}$ on each step. After $\ell$ steps, the random walk $X_1, \ldots, X_\ell$ deviates from the origin by $O(1)$ in \emph{expectation} (but could deviate by much more) and has the property that each step is \emph{deterministically} the same size. The Zeno walk does the complement of this: it deviates from the origin by at most $1$ \emph{deterministically}, but to do this it decreases the size of the $i$-th step by a factor of $1 / \alpha_i$. The key property that we will prove (Proposition \ref{prop:zeno}) is that, although the multiplier $1 / \alpha_i$ can potentially be large, the \emph{expected value} satisfies $O(1 / \alpha_i) = O(1)$ for $i \in [\ell]$. With this intuition in mind, we can now begin the analysis.

Define $Y_i := \ln (1 / (1 - Z_i))$. Rather than analyze the $Z_i$'s directly, we will instead analyze the $Y_i$'s. We will see that the sequence $Y_1, Y_2, \ldots$ behaves similarly to the standard random walk $X_1, X_2, \ldots$ that we described in the previous paragraph (except that (1) $Y_i$ is slightly biased and (2) $Y_i$ can never go below $\ln 0.5$). To make this more precise, the next lemma shows that the random walk $Y_1, Y_2, \ldots$ takes steps of size at most $O(\delta)$ and has bias at most $O(\delta^2)$ per step.

\begin{lemma}
For $i \ge 0$, we have that
\begin{equation}
    |Y_{i + 1} - Y_{i}| = O(\delta)
    \label{eq:sub1}
\end{equation}
deterministically, and that
\begin{equation}
    \Big| \E[Y_{i + 1} - Y_{i} \mid Y_1, \ldots, Y_i, \delta_i]\Big| = O(\delta^2).
    \label{eq:sub2}
\end{equation}
\label{lem:submar}
\end{lemma}
\begin{proof}
Define $$\gamma_i = \frac{\alpha_i \delta_i}{1 - Z_i}.$$
Note that, if $Z_i \ge 0$, then $\gamma_i = \delta_i$, and otherwise $\gamma_i < \delta_i$.  Since $Z_{i + 1} = Z_i \pm (1 - Z_i) \gamma_i$, we have that 
\begin{align*}
Y_{i + 1} & = \ln \left(\frac{1}{1 - Z_i \pm (1 - Z_i) \gamma_i}\right) \\
          & = \ln \left(\frac{1}{1 - Z_i} \cdot \frac{1}{1 \pm \gamma_i}\right) \\
          & = \ln \left(\frac{1}{1 - Z_i}\right)  + \ln \left(\frac{1}{1 \pm \gamma_i}\right) \\
          & = Y_i + \ln \left(\frac{1}{1 \pm \gamma_i}\right).
\end{align*}
By a Taylor approximation, we know that $\ln \left(\frac{1}{1 \pm \gamma_i}\right)$ is within $O(\gamma_i^2)$ of $\pm \gamma_i$. 
That is, $Y_{i + 1}$ can be computed from $Y_i$ by first adding $\pm \gamma_i$ at random to $Y_i$, and then adding/subtracting an additional $O(\gamma_i^2)$. We therefore have that
$$|Y_{i + 1} - Y_{i}| \le \gamma_i + O(\gamma_i^2) \le \delta_i + O(\delta_i^2) \le O(\delta)$$ 
and that 
$$\Big|\E[Y_{i + 1} - Y_{i} \mid Y_1, \ldots, Y_i, \gamma_i] \Big| \le O(\gamma_i^2) \le O(\delta_i^2) \le O(\delta^2).$$

\end{proof}

Using Lemma \ref{lem:submar}, we can now bound $\E[1 / \alpha_\ell]$ for the $\ell = O(1 / \delta^2)$-th step of a Zeno walk:
\begin{proposition}
For $\ell = O(1 / \delta^{2})$, we have $\E[1 / \alpha_\ell] = O(1)$. 
\label{prop:zeno}
\end{proposition}
\begin{proof}
By symmetry, it suffices to show that
$$\E[1 / \alpha_\ell \cdot \mathbb{I}_{Z_\ell \ge 0}] = O(1),$$
where $\mathbb{I}_{Z_\ell \ge 0}$ is 0-1 indicator random variable for the event $Z_{\ell} \ge 0$. 
Note that 
\begin{align*}
    \E[1 / \alpha_\ell \cdot \mathbb{I}_{Z_\ell \ge 0}] & = E[1 / (1 - Z_\ell) \cdot \mathbb{I}_{Z_\ell \ge 0}] \\
                                          & \le  E[1 / (1 - Z_\ell)],
\end{align*}
so we can complete the proof by showing that
\begin{equation}
\E[1 / (1 - Z_\ell)] = O(1).
\label{eq:1minz}
\end{equation}

Let $c$ be a sufficiently large positive constant and define the sequence $X_1, X_2, \ldots$, where
$$X_i = Y_i - i \cdot c \delta^2.$$
This means that $X_{i + 1} - X_i = Y_{i + 1} - Y_i - c \delta^2$, so we can think of the $X_i$'s as being a modification of the $Y_i$'s that eliminates any upward bias that the $Y_i$'s might have (recall by Lemma \ref{lem:submar} that the $Y_i$'s have bias at most $O(\delta^2)$). 

Formally, one can apply  Lemma \ref{lem:submar} to deduce that the $X_i$'s are a supermartingale with bounded differences of $O(\delta)$. That is, by \eqref{eq:sub2} we have  $\E[X_{i + 1} \mid X_1, \ldots, X_i] \le X_i$ (so the $X_i$'s form a supermartingale) and by \eqref{eq:sub1} we have $|X_{i + 1} - X_i| \le O(\delta)$ (so the martingale has bounded differences of $O(\delta)$). 

We can apply Azuma's inequality for supermartingales with bounded differences to deduce the following tail bound. For $k \ge 1$, we have
$$\Pr[X_i \ge \delta k \sqrt{i}] \le e^{-\Omega(k^2)}.$$ 
Unrolling the definition of $X_i$, we get that
$$\Pr[\ln (1 / (1 - Z_i)) \ge \delta k \sqrt{i} + i c \delta^2] \le e^{-\Omega(k^2)}.$$
Plugging in $i = \ell = O(1 / \delta^2)$, we conclude that
$$\Pr[\ln (1 / (1 - Z_\ell)) \ge \Omega(k)] \le e^{-\Omega(k^2)}.$$
This further simplifies to
$$\Pr\left[1 / (1 - Z_\ell) \ge e^{\Omega(k)}\right] \le e^{-\Omega(k^2)},$$
which implies \eqref{eq:1minz}, and completes the proof.
\end{proof}

We conclude the section by generalizing Zeno walks to take place in an arbitrary interval $(\lambda - \epsilon, \lambda + \epsilon)$. This works exactly as before, except that now the Zeno walk begins at $Z_0 = \lambda$; it deterministically stays in the interval $(\lambda - \epsilon, \lambda + \epsilon)$; it sets $\alpha_i = \epsilon - |Z_i - \lambda|$ to be the distance from $Z_i$ to the nearest boundary $\lambda - \epsilon $ or $\lambda + \epsilon $; and then $Z_{i + 1} = Z_i \pm \alpha_i \delta_i$ where $\delta_i \le \delta$ is selected by an adversary. Equivalently, a sequence $\{Z_i\}$ is a Zeno walk in the interval $(\lambda - \epsilon, \lambda + \epsilon)$ if $\{(Z_i -  \lambda) / \epsilon\}$ is a Zeno walk in $(-1, 1)$ (and the two Zeno walks have the same parameter $\delta$ as each other). Thus we get the following generalization of Proposition \ref{prop:zeno}. 
\begin{proposition}
Consider a Zeno walk in $(\lambda - \epsilon , \lambda + \epsilon )$. For $\ell = O(1 / \delta^{2})$, we have $\E[1 / \alpha_\ell] \le O(\epsilon^{-1})$. 
\label{prop:zeno2}
\end{proposition}

\section{The Zeno Embedding: a Data Structure for $m \ge (1 + \epsilon) n$} \label{sec:dense}

In this section, we give a list-labeling solution for $m \ge (1 + \epsilon) n$ that achieves expected cost $O(\epsilon^{-1}\log^{3/2} n)$ per insertion and deletion. We will treat $m \in \mathbb{N}$ and $\epsilon \in (0, 1)$ as being fixed, and we will allow the number $n$ of elements to vary subject to the constraint that $m \ge (1 + \epsilon) n$.
We will also assume without loss of generality that $n$ is at least a sufficiently large positive constant.

We construct and analyze the data structure in three phases. First, we describe a certain type of static construction, which we call the \defn{Zeno embedding}, for how to embed $n$ elements into $m$ slots. Then we show how to dynamize the Zeno embedding in order to efficiently implement \emph{random} insertions/deletions.  Finally, we present one last modification to the Zeno embedding in order to implement arbitrary insertions/deletions efficiently.

\subsection{The Static Zeno Embedding}
The Zeno embedding treats the array as having a simple recursive structure: the \defn{level-0 subproblem} consists of the entire array; and the \defn{level-$i$ subproblems} each consist of either $\lfloor m / 2^i\rfloor$ or $\lceil m / 2^i\rceil$ contiguous slots in the array. 

Each level-$i$ subproblem $S$ is either a \defn{base case} (meaning it does not have child subproblems) or has two recursive children. If $S$ has $q \in \{\lfloor m / 2^i\rfloor, \lceil m / 2^i\rceil\}$ slots, then the children of $S$ have $\lfloor q/2\rfloor$ and $\lceil q / 2 \rceil$ slots, respectively. Here we are taking advantage of the basic mathematical fact that $$\{\lfloor \lfloor m / 2^i\rfloor/2 \rfloor, \lfloor\lceil m / 2^i\rceil /2 \rfloor, \lceil \lfloor m / 2^i\rfloor / 2 \rceil, \lceil \lceil m / 2^i\rceil / 2 \rceil\} \subseteq \{ \lfloor m / 2^{i + 1}\rfloor, \lceil m / 2^{i + 1}\rceil\}.$$

For each level-$i$ subproblem $S$, define $|S|$ to be the number of elements stored in that subproblem, and define the \defn{density} $\mu_S$ of the subproblem to be 
$$\mu_S = \frac{|S|}{n / 2^i}.$$

Note that in the definition of $\mu_S$, the denominator is the average number of elements per level-$i$ subproblem, which means that $\mu_S$ can be greater than 1. In fact, we will guarantee deterministically that $\mu_S \in [1 - \epsilon / 2, 1 + \epsilon / 2]$. The upper bound will ensure correctness (i.e., that no subproblem overflows), and the lower bound will ensure that every pair of consecutive elements are within $O(1)$ slots of each other.

We can now describe how to implement a given level-$i$ subproblem $S$. Define
$$\alpha_S = \epsilon/2 - |1 - \mu_s|$$
to be the distance between $\mu_S$ and the nearest boundary $\{1 - \epsilon / 2, 1 + \epsilon / 2\}$.  Let $x_1, \ldots, x_{|S|}$ denote the elements of $S$ in sorted order. Define the \defn{pivot candidate set} for $S$ to be
$$C_S = \left\{x_i \Bigm| \frac{|S|}{2} - \frac{n}{2^i} \cdot \frac{\alpha_S}{\sqrt{\log n}} \le i \le \frac{|S|}{2} +\frac{n}{2^i} \cdot \frac{\alpha_S}{\sqrt{\log n}} \right\}.$$
Roughly speaking, $C_S$ consists of the elements representing the middle $\Theta(\alpha_S / \sqrt{\log n})$-fraction of the subproblem.

If $|C_S| \le 4$, we declare $S$ to be a \defn{base case}, and we spread the elements of $S$ evenly across its slots. 
Otherwise, we define the \defn{pivot} $p_S$ for $S$ to be an element of $C_S$ chosen uniformly at random. The elements $x_i \le p_S$ are recursively placed in $S$'s left child, and the elements $x_i > p_S$ are recursively placed in $S$'s right child.

Later on, when we discuss the \emph{dynamic} Zeno embedding, we will see several ways that one can implement the random choice of $p_S$. For concreteness, we will mention one natural approach here: define $h_0, h_1, h_2, \ldots, h_{O(\log n)}$ to be an independent sequence of hash functions\footnote{Technically, our data structure does not necessarily have access to the internal values of elements, so it cannot compute a hash $h_i(x)$ of any given element. However, we can simulate a hash function $h_i$ by assigning each element $x$ a random value $h_i(x)$ when the element is inserted.} where each $h_i$ maps each element to a uniformly random real number in $[0, 1]$, and set
$$p_S = \operatorname{argmin}_{x \in C_S} h_i(x).$$

The key property of the Zeno embedding is that if we perform a random walk down the recursive tree, then the densities $\mu_S$ that we encounter form an $O(\log n)$-step Zeno walk in the interval $[1 - \epsilon / 2, 1 + \epsilon / 2]$:

\begin{lemma}
Fix any outcomes for the hash functions $h_0, h_1, h_2, \ldots$. Consider a random walk $S_0, S_1, S_2, \ldots, S_\ell$ down the recursion tree, where each $S_{i + 1}$ is a random child of $S_i$, and $S_\ell$ is a base-case subproblem. Then the sequence $\{\mu_{S_i}\}_{i = 1}^\ell$ is a Zeno walk on $[1 - \epsilon / 2, 1 + \epsilon / 2]$ with $\delta = O(1 / \sqrt{\log n})$.
\label{lem:zenointree}
\end{lemma}
\begin{proof}
Recall that a Zeno walk on $[1 - \epsilon / 2, 1 + \epsilon / 2]$ is any walk $Z_0, Z_1, \ldots$ that starts at $1$ and takes the following form: each step $Z_{i + 1} - Z_{i}$ is randomly $\pm \alpha_i \delta_i$ for some $\delta_i \le \delta$ (that may be chosen by an adversary) and where $\alpha_i =  \epsilon/2 - |1 - Z_i|$.  Or, equivalently, each step $Z_{i + 1} - Z_{i}$ is randomly $\pm \beta_i$ for some $\beta_i \le \delta \left( \epsilon/2 - |1 - Z_i| \right)$. 

Consider a non-base-case subproblem $S_i$, and let $A$ and $B$ be the child subproblems of $S_i$. By construction, 
$$\big||A| - |B|\big| = O\left(\frac{n}{2^i} \cdot \frac{\alpha_S}{\sqrt{\log n}}\right).$$
Since $|A| + |B| = |S_i|$, we have that $\mu_A + \mu_B = 2 \mu_{S_i}$ and 
$$|\mu_A - \mu_B| =  \frac{\big||A| - |B|\big|}{n / 2^{i + 1}} =  O\left(\frac{\alpha_{S_i}}{\sqrt{\log n}}\right).$$
Thus, since $S_{i + 1}$ is randomly one of $A$ or $B$, we have that $\mu_{S_{i + 1}}$ is randomly one of
$$\mu_{S_i} + \beta_i  \text{ or } \mu_{S_i} - \beta_i,$$
where 
$$\beta_i = |\mu_A - \mu_B| / 2 =  O\left(\frac{\alpha_{S_i}}{\sqrt{\log n}}\right) = O\left( \delta \left(\epsilon / 2 - |1 - \mu_s|\right) \right).$$
Thus the sequence $\{\mu_s\}$ is a Zeno walk on $[1 - \epsilon/2, 1 + \epsilon/2]$ with $\delta = O(1 / \sqrt{\log n})$. 

For clarity, we remark that the definition of the Zeno walk includes an adaptive adversary who chooses $\delta_i<\delta$. The adversary for the Zeno walk in this lemma simply chooses a pivot uniformly at random from the pivot candidate set, which determines $\delta_i$.
\end{proof}

The reason that Lemma \ref{lem:zenointree} is important is that it allows for us to bound the quantities $\alpha_S^{-1}$. Indeed, we use Proposition \ref{prop:zeno2} to prove the following inequality.
\begin{lemma}
Let $\mathcal{S}_i$ be the set of level-$i$ subproblems. Then
$$\frac{1}{2^i} \sum_{S \in \mathcal{S}_i} \alpha_S^{-1} = O(\epsilon^{-1}).$$
\label{lem:alpha}
\end{lemma}
\begin{proof}
Fix any outcomes for the hash functions $h_0, h_1, h_2, \ldots $. Consider a random walk $S_0, S_1, S_2, \ldots, S_\ell$ down the recursion tree, where each $S_{i + 1}$ is a random child of $S_i$, and $S_\ell$ is a base-case subproblem. Lemma \ref{lem:zenointree} tells us that $\{\mu_{S_i}\}_{i = 1}^\ell$ is a Zeno walk on $[1 - \epsilon / 2, 1 + \epsilon / 2]$ with $\delta = O(1 / \sqrt{\log n})$ (and, moreover, $\alpha_{S_i}$ corresponds to $\alpha_i$ in the Zeno walk). 

For $i \in [0, \log m]$, define $\overline{\alpha}_i$ to be $\alpha_{S_i}$ if $S_i$ exists and $0$ otherwise (i.e., if $i > \ell$). Proposition \ref{prop:zeno2} tells us that, for each $i \in [0, \log m]$,
\begin{equation}
    \E[1 / \overline{\alpha}_i] = O(\epsilon^{-1}).
    \label{eq:alpha0}
\end{equation}
On the other hand, each level-$i$ subproblem has probability exactly $1 / 2^i$ of being $S_i$. Thus
\begin{equation}
    \E[1 / \overline{\alpha}_i] = \frac{1}{2^i} \sum_{S \in \mathcal{S}_i} \alpha_S^{-1}.
    \label{eq:alpha1}
\end{equation}
Combined, \eqref{eq:alpha0} and \eqref{eq:alpha1} imply the lemma.
\end{proof}

It is interesting to note that, whereas Lemma \ref{lem:zenointree} is a statement about random walks, Lemma \ref{lem:alpha} is a \emph{deterministic} bound on the $\alpha_S^{-1}$s, even though it uses a probabilistic argument to derive the bound. 

Lastly, we also need to explicitly show that no subproblem ever overflows:

\begin{restatable}{lemma}{overflow}
Each level-$i$ subproblem $S$ satisfies $|S| \le \lfloor m / 2^i \rfloor$. 
\label{lem:overflow}
\end{restatable}

This lemma is a technicality that is essentially immediate from the fact that each subproblem $S$ has density $\mu_S \le 1 + \epsilon / 2$. The only difficulty in the proof comes from the necessity to carefully handle floors/ceilings. We defer the proof to Appendix \ref{app:over}.

\subsection{Dynamizing the Zeno Embedding}

We now describe a dynamic version of the Zeno embedding; we will treat $m$ and $\epsilon$ as fixed, and allow $n$ to vary subject to the constraint that $n \ge (1 + \epsilon) m$. 

We note that, in this section we will focus on analyzing \emph{random} insertions/deletions, that is, an insertion/deletion that is performed at a random rank (in an array with arbitrary contents). Our solution will be history independent, and we will see in the next subsection that this allows the random-rank assumption to be removed.

\paragraph{Implementing insertions and deletions.} To implement an insertion/deletion in the Zeno embedding, we simply update the embedding to account for the element being added/removed. More concretely, we can implement an insertion/deletion of an element $x$ as follows. We will describe the process recursively, focusing on how to insert/delete $x$ into a given level-$i$ recursive subproblem $S$. The insertion/deletion of $x$ may change the values of $\mu_S, \alpha_S, C_S,$ and $p_S$. Note that the values of $C_s$ and $p_S$ can change regardless of whether the insertion/deletion of $x$ takes place in the candidate set. If it changes the pivot $p_S$, or if $S$ is a base-case, then we implement the insertion/deletion by rebuilding the entire subproblem from scratch, incurring a cost of $O(n / 2^i)$. Otherwise, we recursively insert/delete $x$ into either the left child (if $x \le p_S$) or the right child (if $x > p_S$). Once the insertion/deletion is complete, the Zeno embedding will be the same as if it were constructed from scratch on the current set of elements.

As described in the static Zeno embedding, there are multiple ways to implement randomly choosing a pivot. One way is to use the hash functions $h_i$ described in the previous subsection. This means that a level-$i$ subproblem $S$ being inserted/deleted into gets rebuilt if $\operatorname{argmin}_i \{h_i(x) \mid x \in C_S\}$ is changed by the insertion/deletion. We note that, in this construction, the hash functions are fixed at the very beginning and are never resampled (even when subproblems are rebuilt). 

Another way to implement the random choice of pivot is to use reservoir sampling \cite{AragonSe89, SeidelAr96, vitter1985random, li1994reservoir, BenderBeJo16}. This means that, when a subproblem is first built (or rebuilt), it picks a random $x \in C_S$ to be the pivot; whenever an element $x$ is added to $C_S$, it has probability $1 / |C_S \cup \{x\}|$ of becoming the pivot; and whenever an element $x$ is removed from $C_S$, if $x$ was the pivot, then a random element in $C_S \setminus \{x\}$ is chosen as the new pivot. Like the hashing method, reservoir sampling maintains as an invariant that each candidate in $C_S$ is equally likely to be the pivot. 

Each of the two methods (hashing and reservoir sampling) have their own benefits: reservoir sampling can be used to immediately obtain an algorithm in the RAM-model that has the same asymptotic running time as its list-labeling cost, while hashing, on the other hand, ensures that the embedding is deterministic after fixing the hash functions. In our formal arguments, we use the hash function method, but this can easily be replaced with reservoir sampling.

\paragraph{Analyzing a random insertion/deletion.} To begin analyzing the dynamic Zeno embedding, we observe that, by construction, the dynamic Zeno embedding is insertion/deletion symmetric and  history independent.

\begin{observation}
The dynamic Zeno embedding is insertion/deletion symmetric and history independent.
\end{observation}

Due to the insertion/deletion symmetry, the expected cost of a random insertion on an array with $n$ elements is the same as the expected cost of a random deletion on an array with $n+1$ elements. Thus we need only analyze the expected cost of a random deletion.

We will analyze the probability that the deletion of an element $x$ causes the rebuild of a subproblem. More precisely, we say that a subproblem $S$ is \defn{rebuilt} if the pivot of $S$ changes, while the pivots of all of the ancestors of $S$ do not change.

Next, we will prove that, if we delete an element $x$, and $S$ is the level-$i$ subproblem that contains $x$, then the probability that $S$ is rebuilt is $O(|C_S|^{-1})$.
\begin{lemma}
If an element $x$ is deleted from a subproblem $ S $, then $S$ is rebuilt with probability
$$O\left(|C_S|^{-1}\right).$$
\label{lem:rebuildprob}
\end{lemma}
\begin{proof}
If $S$ is a base-case subproblem, either before or after the deletion, then $|C_S| = O(1)$, and the lemma is trivial. Now, suppose $S$ is not a base-case subproblem. 

Let $C_S$ denote the pivot candidate set prior to the deletion of $x$, and let $\overline{C}_S$ denote the pivot candidate set after the deletion. Each time that we add/remove an element to/from $C_S$, the probability that $p_S = \operatorname{argmin}_{x \in C_S} h_i(x)$ changes is $\Theta(1 / |C_S|)$. It therefore suffices to show that $C_S$ and $\overline{C}_S$ have a symmetric difference of at most $O(1)$ elements. 

We can think of the transformation of $ C_S $ into $\overline{C}_S$ as taking place in three steps. First we update
$$\alpha_S = \epsilon/2 - \left|1 - \frac{|S|}{n / 2^i}\right|$$
to become
$$\alpha_S = \epsilon/2 - \left|1 - \frac{|S|-1}{n / 2^i}\right|.$$
This changes $\alpha_S$ by at most $\pm \frac{1}{n / 2^i}$, which changes the set
\begin{equation}
    C_S = \left\{x_i \Bigm| \frac{|S|}{2} - \frac{n}{2^i} \cdot \frac{\alpha_S}{\sqrt{\log n}} \le i \le \frac{|S|}{2} +\frac{n}{2^i} \cdot \frac{\alpha_S}{\sqrt{\log n}} \right\}
    \label{eq:C}
\end{equation}
by at most $O(1)$ elements. Second, we replace $|S|$ in \eqref{eq:C} with $|S| - 1$. This again changes the set $C_S$ by at most $O(1)$ elements. Third, we remove the element $x$; if $x = x_j$ for some $j$, then the removal of $x$ has the effect of decrementing the index of each $x_i$ with $i \ge j$. This again changes $C_S$ by at most $O(1)$ elements. 

Combined, the three steps complete the transformation of $C_S$ into $\overline{C}_S$, meaning that $\overline{C}_S$ and $C_S$ have a symmetric difference of $O(1)$ elements, as desired.
\end{proof}

Lemma \ref{lem:rebuildprob} immediately implies a bound on the expected cost incurred from rebuilding $S$.
\begin{lemma}
If an element $x$ is deleted from a level-$i$ subproblem $ S $, the expected cost incurred from possibly rebuilding $S$ is
$$O\left(\frac{n/2^i}{|C_S|}\right).$$
\label{lem:rebuildcost}
\end{lemma}
\begin{proof}
A rebuild of $S$ costs $\Theta(n / 2^i)$. Thus the lemma follows from Lemma \ref{lem:rebuildprob}. 
\end{proof}

Observe that, by design, $$\frac{n/2^i}{|C_S|} = O(\alpha_S^{-1} \sqrt{\log n}).$$ This is where Lemma \ref{lem:alpha} comes into play: it tells us that even though $\frac{n / 2^i}{|C_S|}$ may be large for some subproblems $S$, it cannot be consistently large across all subproblems. Using this, we can analyze the expected cost to delete a random element. 

\begin{lemma}
The expected cost to delete a random element $x$ from the Zeno embedding is $O(\epsilon^{-1}\log^{3/2} n)$.
\label{lem:randomdelete}
\end{lemma}
\begin{proof}
Let $\mathcal{S}_i$ denote the set of level-$i$ subproblems (prior to the deletion). Each $S \in \mathcal{S}_i$ contains
$\Theta(n / 2^i)$ elements, so 
$$\Pr[x \in S] = \Theta\left(\frac{1}{2^i}\right).$$
If $x \in S$, then we have by Lemma \ref{lem:rebuildcost} that $S$ incurs expected rebuild cost
$$O\left(\frac{n / 2^i}{|C_S|}\right) = O(\alpha_S^{-1} \sqrt{\log n}).$$
The expected cost from rebuilds in the $i$-th level of recursion is therefore at most
$$O\left(\sum_{S \in \mathcal{S}_i} \frac{1}{2^i} \cdot \alpha_S^{-1} \sqrt{\log n}\right),$$
which by Lemma \ref{lem:alpha} is at most
$$O\left(\epsilon^{-1} \sqrt{\log n}\right).$$
Summing over the $O(\log n)$ levels of recursion, the total expected cost of the deletion is $O(\epsilon^{-1}\log^{3/2} n)$.
\end{proof}

Due to the previously described symmetry between insertions and deletions, the same lemma is true for insertions.
\begin{lemma}
The expected cost to insert an element $x$ with a random rank in $\{1, 2, \ldots, n + 1\}$ into the Zeno embedding is $O(\epsilon^{-1}\log^{3/2} n)$.
\label{lem:randominsert}
\end{lemma}

\subsection{Achieving a Bound on Arbitrary Insertions/Deletions.}
So far, we have only analyzed random insertions/deletions. At first glance, this may seem like an insignificant accomplishment. (Indeed, it is already known that random insertions/deletions can be supported in $O(\epsilon^{-1})$ amortized time per operation~\cite{DBLP:journals/mst/BenderFM06}.) 

What makes the Zeno embedding special is that it is history independent. We will now show how to reduce the list-labeling problem (with \emph{arbitrary} insertions/deletions) to the problem of constructing an insertion/deletion-symmetric history-independent embedding that supports efficient \emph{random} insertions/deletions. 

Within any history-independent data structure, the expected cost to perform a deletion at rank $r$ on an array of size $m$ containing $n$ elements can be expressed by a \defn{cost function} $T(m,n,r)$ only dependent on $m$, $n$ and $r$. Moreover, if the data structure is insertion/deletion symmetric, then the same cost function $T$ expresses the expected cost for an insertion; specifically, the expected cost to perform an insertion at rank $r$ on an array of size $m$ containing $n$ elements is $T(m,n-1,r)$.

To reduce from the arbitrary insertion/deletion case to the random insertion/deletion case, we will show that given any (insertion/deletion-symmetric) history-independent algorithm $\mathcal{A}$ with cost function $T(m,n,r)$, we can construct a history-independent algorithm $\mathcal{B}$ with cost function $T'(m,n,r)$ such that for each individual rank $r$, the cost $T'(m,n,r)$ is upper bounded by the \emph{average} of the costs $T(m,n,r)$ across all ranks (up to constant factors).

\begin{lemma}
Suppose there is an insertion/deletion-symmetric history-independent algorithm $\mathcal{A}$ whose cost is determined by a function $T(m, n, r)$. Then we can construct a new insertion/deletion-symmetric history-independent algorithm $\mathcal{B}$ with cost function $T'(m, n, r)$ satisfying
$$T'(m, n, r) = O\left(\frac{1}{m + 1}  \sum_{j = 1}^{2m} T(2m, m + n, j)\right)$$
for all $r$. 
\label{lem:randreduction}
\label{lem}
\end{lemma}
\begin{proof}
Fix a history-independent algorithm $\mathcal{A}$. We will construct a history-independent algorithm $\mathcal{B}$. We will describe the behavior of the algorithm $\mathcal{B}$ on an array of size $m$ with an arbitrary sequence $\mathcal{S}$ of insertions/deletions. 

To do so, we will construct from $\mathcal{S}$ an input to $\mathcal{A}$. The input to $\mathcal{A}$ is an array of size $2m$ with the following insertion/deletion sequence. First we insert $m$ dummy elements as follows. Let $q$ be a uniformly random integer in $[0, m]$. Insert $q$ dummy elements that are treated as taking infinitely small values (i.e., $-\infty$), and insert $m-q$ dummy elements that are treated as taking
 infinitely large values (i.e., $\infty$). Now, execute the sequence $\mathcal{S}$.
 
 Now, define $\mathcal{B}$ as the algorithm that behaves identically to $\mathcal{A}$ on $\mathcal{A}$'s subarray $[q,q+m]$  (that is, $\mathcal{A}$'s subarray from the $q^{th}$ slot to the $q+m^{th}$ slot), ignoring the dummy elements. That is, for all $i$, after the $i^{th}$ insertion from $\mathcal{S}$, the subarray $[q,q+m]$ of $\mathcal{A}$'s array with the dummy elements removed, is identical to $\mathcal{B}$'s array.
 
 We note that $\mathcal{B}$ is well defined in the sense that all elements of $\mathcal{S}$ always appear in $\mathcal{A}$'s subarray $[q,q+m]$. This is simply due to the existence of the dummy elements in $\mathcal{A}$'s array.

Now let us bound the expected cost $T'(m,n,r)$ for $\mathcal{B}$ to perform a deletion at rank $r$. This corresponds to a deletion at rank $r + q$ in $\mathcal{A}$, which has cost $T(2m, m + n, r + q)$. Notice, however, that $r + q$ is a random element in $\{r, r + 1, \ldots, r + m\}$. Thus,$$T'(m,n,r)=\frac{1}{m + 1} \sum_{j = r}^{r + m} T(2m, m + n, j),$$
which in turn is at most
$$O\left(\frac{1}{m + 1}  \sum_{j = 1}^{2m} T(2m, m + n, j)\right).$$
\end{proof}

In the case where $\mathcal{A}$ is the Zeno embedding, we refer to $\mathcal{B}$ as the \defn{shifted Zeno embedding}. Now, we are ready to put everything together and prove our main theorem, that the shifted Zeno embedding incurs expected cost $O(\epsilon^{-1} \log^{3/2} n)$ per insertion/deletion.

\begin{theorem}
Let $\epsilon \in (0, 1)$, and suppose $m \ge (1 + \epsilon) n$, where $m$ is a static value while $n$ changes dynamically. The shifted Zeno embedding on an array of size $m$ with $n$ elements incurs expected cost $O(\epsilon^{-1} \log^{3/2} n)$ per insertion/deletion.
\label{thm:shiftedzeno}
\end{theorem}
\begin{proof}
Let $T(m,n,r)$ be the cost function associated with the Zeno embedding, and let $T'(m,n,r)$ be the cost function associated with the shifted Zeno embedding. From Lemma \ref{lem:randreduction}, we know that \begin{equation} T'(m, n, r) = O\left(\frac{1}{m + 1}  \sum_{j = 1}^{2m} T(2m, m + n, j)\right).\end{equation}\label{eqn:av}The right side of Equation \ref{eqn:av} is within a constant factor of the average value of $T(2m, m + n, j)$ over all ranks $j$. Thus, it is within a constant factor of the expected value of $T(2m, m + n, j)$ where $j$ is chosen uniformly at random over all ranks, which we know from Lemmas \ref{lem:randomdelete} and \ref{lem:randominsert}, is $O(\epsilon^{-1}\log^{3/2}n)$. Thus, $T'(m, n, r)=O(\epsilon^{-1}\log^{3/2}n)$, as desired.
\end{proof}

The following corollary follows immediately by applying Theorem \ref{thm:shiftedzeno} to an $n(1+\epsilon)$ sized
subarray of a linearly sized array for any $\epsilon<1$. 

\begin{corollary}
There exists a list-labeling algorithm for an array of size $m=n(1+\Theta(1))$ with expected cost $O(\log^{3/2} n)$ per insertion/deletion.
\label{thm:linearcase}
\end{corollary}

We can also use the theorem to bound the total cost to insert into every slot in an array. 

\begin{corollary}
There exists a list-labeling algorithm to fill an array of size $m$ from empty to full with expected total cost $O(m\log^{2.5} m)$.
\label{cor:fill_array}
\end{corollary} 

\begin{proof}
We will apply a shifted Zeno embedding in $\Theta(\log m)$ phases, using an $\epsilon_i$, defined below, for phase $i$ and rebuilding the array between phases. The first phase consists of the first $m/2$ insertions, and each phase inserts half as many elements as the preceding phase. This continues until $n > m - \log m$, at which point the final phase consists of inserting the remaining at most $\log m$ elements.

More precisely, let $k =\lceil \log m - \loglog m \rceil$, and define 
\[n_i=\frac{m(2^i-1)}{2^i} \hspace{1em}\text{for }i=0,1,\dots,k,\]
and $n_{k+1} = m$.

Items are inserted by ranks, specified by $r_1, \ldots,r_m$, so that for example, since the first insertion is into an empty array, $r_1 = 1$.  Phase $P_i$ is defined by the insertions $r_j$ with $j\in (n_{i-1},n_i]$.  We define $\epsilon_i = (2^i-1)^{-1}$ for $i>1$, and $\varepsilon_1 = \frac{m-1}{m}$.

Let $C(P_i)$ denote the expected total cost of the insertions in phase $P_i$. For $i>1$ and for all $j\in (n_{i-1},n_i]$,
\[ (1+\epsilon_i)j \leq (1+\epsilon_i)n_i = \left(1 + \frac{1}{2^i-1}\right)\left(\frac{m(2^i-1)}{2^i}\right) = m.\]
Similarly, in phase $P_1$, we have
\[ (1+\epsilon_i)j \leq \left(1 + \frac{m-1}{m}\right)\cdot\frac{m}{2} \leq m.\]
Therefore, we can apply Theorem \ref{thm:shiftedzeno} to say that for all $i$, an insertion during phase $P_i$ incurs expected cost
\[O(\epsilon_i^{-1}\log^{3/2} n) = O(2^i\log^{3/2}m),\]
and \[C(P_i) = O\left(2^i\cdot\frac{m}{2^i}\cdot\log^{3/2}m\right) = O(m\log^{3/2}m).\]
Summing over the first $k = O(\log m)$ phases, this gives expected total insertion cost $O(m\log^{2.5}m)$.

By construction, the final phase has at most $\log m$ insertions, and thus has total expected cost $O(m\log m)$.
Finally, since the total number of elements in the array is bounded by $m$, the rebuilds between phases incur total cost $O(m\log m)$, completing the proof.

\end{proof}

We conclude the section with a remark.
\begin{remark}
Many applications of list labeling require that, if $m = \Theta(n)$, then the number of empty slots between any two consecutive elements is at most $O(1)$. The Zeno embedding satisfies this property by design, since each subproblem has density at least $1 - \epsilon / 2$. The shifted Zeno embedding therefore also satisfies the same property. 
\end{remark}

\section{A Lower Bound for History-Independent Solutions}\label{sec:lower}

The shifted Zeno embedding (Theorem~\ref{thm:shiftedzeno}) has the property that it is \defn{history independent}, meaning that the state of the data structure does not reveal any information about the history of insertions/deletions. In this section, we prove that the $\epsilon^{-1} \log^{3/2} n$ bound achieved by the shifted Zeno embedding is, in fact, optimal  for history-independent data structures.

The main result of this section will be the following lower bound:
\begin{theorem}
Consider any history-independent list-labeling data structure. Let $m$ be the size of the array and let $n = (1 - \epsilon)m$, where $\epsilon$ is at most some small positive constant and is at least $m^{-1/3}$. The expected cost to insert an element with a random rank in $\{1, 2, \ldots, n + 1\}$ and then delete the element with rank $n + 1$ is $\Omega(\epsilon^{-1} \log^{3/2} n)$. 
\label{thm:lower} 
\end{theorem}

Throughout the rest of the section, let $ c $ be a large constant, and assume that $m$ is sufficiently large as a function of $ c $. Let $m^{-1/3} \le \epsilon \le 1 / c$, and set $n = (1 - \epsilon)m$. We shall consider sequences of insertions/deletions, where each insertion is into an array of $n$ elements and each deletion is from an array of $n + 1$ elements.

To aid in the proof of Theorem \ref{thm:lower}, let us take a moment to establish several definitions and conventions. Let $J = \{2, 4, 8, \ldots, 2^{\lfloor \log m \rfloor - 2} \}$. For each $j \in [m]$, define a \defn{$j$-block} to be a block of $j$ consecutive slots in the array, allowing for wrap-around (so there are $m$ possible $j$-blocks). 

Define the \defn{density} of a $j$-block to be $k / j$, where $k$ is the number of elements in the $j$-block. Call a $j$-block \defn{live} if it has density at least $1 - c\epsilon$, and dead otherwise. Note that this definition of density is slightly different from that used for recursive subproblems in the upper-bound section (Section \ref{sec:dense}) in that we define the density to be between $0$ and $1$---this difference will make the algebraic manipulation cleaner in several places.

For each $j \in J$, define the \defn{imbalance} of a $j$-block to be $|\mu_1 - \mu_2|$, where $\mu_1$ is the density of the first $j / 2$ slots in the block, and $ \mu_2$ is the density of the final $j / 2$ slots in the block. Define the \defn{adjusted imbalance} $\Delta(x)$ of a $j$-block $x$ to be the block's imbalance if the block is live, and $0$ if the block is dead. Finally, define the \defn{boundary set} $B(x)$ to be the set of up to three elements in positions $\{1, j/2, j\}$ of $x$.

For a given array configuration $A$, define $\Delta_j(A)$ to be the average adjusted imbalance across all $j$-blocks. Finally, define $\Delta_j = \E_{A \sim \mathcal{C}_{n, m}}[\Delta_j(A)]$.

We will split the proof of Theorem \ref{thm:lower} into two key components. Section \ref{sec:imbalances} proves the following combinatorial bound, which holds deterministically for any array configuration.
\begin{proposition}
For any array-configuration $A$ with $n = (1 - \epsilon) m$ elements,
$$\sum_{j \in J} (\Delta_j(A))^2 = O(\epsilon^2).$$
\label{prop:imbalances}
\end{proposition}

Section \ref{sec:imbalances} also uses Cauchy-Schwarz to arrive at the following corollary.
\begin{corollary}
For any array-configuration $A$ with $n = (1 - \epsilon) m$ elements,
$$\frac{1}{|J|} \sum_{j \in J} \Delta_j(A) = O(\epsilon / \sqrt{\log n}).$$
\label{cor:imbalances}
\end{corollary}

Section \ref{sec:costlower} then gives a lower bound in terms of the $\Delta_i$'s on the expected cost that any history-independent data structure must incur. 

\begin{proposition}
Suppose $n = (1 - \epsilon)m$, where $m^{-1/3} \le \epsilon \le 1/c$ and $c$ is some sufficiently large positive constant. Suppose we perform an insertion at a random rank $r \in \{0, \ldots, n + 1\}$ and then delete the element with rank $n + 1$. The expected total cost of the insertion/deletion is at least
$$\Omega\left(\sum_{j \in J} \frac{1}{\Delta_j + 1/j}\right).$$
\label{prop:costlower}
\end{proposition}

Note that the expected cost in Proposition \ref{prop:costlower} is with respect to the randomness introduced by both the random rank $r$ and the randomness in the history-independent data structure.

Intuitively, the above results tell us that any optimal history-independent data structure must behave a lot like the Zeno embedding. Indeed, Corollary \ref{cor:imbalances} tells us that, no matter how we configure our array $A$, it is impossible to achieve imbalances that are consistently $\omega(\epsilon / \sqrt{\log n})$---so, if our goal is to maximize the imbalances in our array, we can't hope to do any better than the Zeno embedding already does. Proposition \ref{prop:costlower} then tells us that small imbalances are necessarily expensive to maintain (and, in fact, the asymptotic relationship between cost and imbalance is the same as the one achieved by the Zeno embedding). Combining the propositions, we can prove the theorem as follows.

\begin{proof}[Proof of Theorem \ref{thm:lower}]
By Corollary \ref{cor:imbalances}, we have
\begin{equation}
\frac{1}{|J|} \sum_{j \in J} \Delta_j = \E_{A \sim \mathcal{C}_{n, m}}\left[\frac{1}{|J|} \sum_{j \in J} \Delta_j(A) \right] \le O\left(\epsilon  /\sqrt{\log n}\right).
\label{eq:cauchy}
\end{equation}
By Proposition \ref{prop:costlower}, the the expected cost of the insertion/deletion is at least
\begin{equation}
\Omega\left(\sum_{j \in J} \frac{1}{\Delta_j + 1/j}\right).
\label{eq:cl}
\end{equation}
If $\Delta_j \le 1 / \sqrt{n}$ for any $j \ge \sqrt{n}$, then \eqref{eq:cl} becomes $\Omega(\sqrt{n}) \ge \Omega(\epsilon^{-1} \log^{3/2} n)$ (since $\epsilon \ge \Omega(n^{-1/3})$), and we are done. On the other hand, if $\Delta_j \ge 1/\sqrt{n}$ for all $j \ge \sqrt{n}$, then \eqref{eq:cl} is at least 
\begin{equation}
\Omega\left(\sum_{j \in J, j \ge \sqrt{n}} \frac{1}{\Delta_j}\right).
\label{eq:jsum}
\end{equation}
By \eqref{eq:cauchy}, we know that at least half of the $\Delta_j$'s in the above sum satisfy $\Delta_j = O(\epsilon / \sqrt{\log n})$.
Thus the expected cost comes out to at least
$$\Omega(\epsilon^{-1} \log^{3/2} n).$$
\end{proof}

\subsection{Proof of Proposition \ref{prop:imbalances}}
\label{sec:imbalances}

Although Proposition \ref{prop:imbalances} is a deterministic statement, we will prove it with a probabilistic argument.

For $j \in J$, define the  \defn{children} of a $ j $-block to be the $ j/2$-blocks consisting of the first and last $ j/2 $ slots of the block, respectively. Also, let $j^* = 2^{\lfloor \log m \rfloor - 2}$ 
be the largest element of $J$.

For a $j$-block $x$ with density $\mu$, define the \defn{potential} $\phi(x)$ to be 
$$\phi(x) = \begin{cases} 0 & \text{ if }x \text{ is dead} \\
                        (\mu - (1 - c\epsilon))^2 & \text{ otherwise}.
\end{cases}$$
For a random $j$-block $x$, one should think of $\phi(x)$ as measuring something similar to (but not quite equal to) the variance of $\mu$. The key differences between what $\phi$ and variance measure is that (1) $\phi$ evaluates directly to $0$ on any $j$-block $x$ that is dead (i.e., has density less than $1 - c\epsilon$), and (2) $\phi$ examines the square of the distance between $\mu$ and the death-threshold $1 - c \epsilon$, rather than the square of the distance from $\mu$ to $\E[\mu] = 1 - \epsilon$.

Note that
$0 \le \phi(x) \le O(\epsilon^2)$ deterministically. On the other hand, we will now see how to relate the expected potential $\phi(x)$ of a random $1$-block to the quantity $\sum_{j \in J} (\Delta_j(A))^2$. 

\begin{lemma}
Let $x$ be a random $1$-block. Then
$$\E[\phi(x)] = \Omega\left(\sum_{j \in J} (\Delta_j(A))^2\right).$$
\label{lem:phia}
\end{lemma}
\begin{proof}
Let $ x_0 $ be a random $j^*$-block, and for $i \in [\log j^*]$, let $x_i$ be a random child of $x_{i - 1}$. This means that each $x_i$ is itself a random $2^{\log j^* - i}$-block, and that $x := x_{\log j^*}$ is a random $1$-block. 
Define $\mu_i$ to be the density of $x_i$.

We will argue that
\begin{equation}
\E[\phi(x_i) - \phi(x_{i - 1})] = \Omega\left((\Delta_{j^* / 2^i}(A))^2\right).
\label{eq:pd}
\end{equation}
This would imply that
$$\E[\phi(x)] = \E[\phi(x_0)] + \sum_i \E[\phi(x_i) - \phi(x_{i - 1})] = \Omega\left(\sum_{j \in J} (\Delta_j(A))^2\right),$$
as desired. 

For the rest of the proof, consider some $\phi_i$ and set $j = j^* / 2^i$. We claim that with probability at least $1 - 1 / c \ge 0.9$, $x_{i - 1}$ is live. Indeed, in expectation at most an $\epsilon$ fraction of the slots in $x_{i - 1}$ are free, so by Markov's inequality the probability that more than a $c\epsilon $ fraction of the slots in $x_{i - 1}$ are free is at most  $1/c \le 0.1$.  

If we condition that $x_{i - 1}$ is live, then its imbalance $\Delta$ satisfies $\E[\Delta] = \Theta(\Delta_{j^* / 2^{i - 1}}(A))$. Furthermore, if $x_{i - 1}$ is live, then we have that $\phi(x_i)$ is randomly one of
$$(\sqrt{\phi(x_{i - 1})} + \Delta/2)^2$$
or
$$(\max\{0,\sqrt{\phi(x_{i - 1})} - \Delta/2\})^2.$$
Note that the average of these is
\begin{align*}
& 0.5 (\sqrt{\phi(x_{i - 1})} + \Delta/2)^2 + 0.5
(\max\{0,\sqrt{\phi(x_{i - 1})} - \Delta/2\})^2.
\end{align*}
If $0<\sqrt{\phi(x_{i - 1})} - \Delta/2$, this average is 
\begin{align*}
\phi(x_{i-1})+\Delta^2/4.
\end{align*}
On the other hand, if $0\ge\sqrt{\phi(x_{i - 1})} - \Delta/2$, this average is
\begin{align*}
&0.5 \cdot  \phi(x_{i - 1})+\Delta \sqrt{\phi(x_{i - 1})}/2+\Delta^2/8\\
&\ge 1.5 \cdot \phi(x_{i - 1})+\Delta^2/8.
\end{align*}
Thus, in either case, this average is \[\ge\phi(x_{i - 1})+\Delta^2/8.\]
It follows that
\begin{align*}
\E[\phi(x_i) \mid x_{i - 1} \text{ live}]& \ge\E[\phi(x_{i - 1}) \mid x_{i - 1} \text{ live} ] + \cdot \E[\Delta^2\mid x_{i - 1} \text{ live}]/8 \\
       					& \ge \E[\phi(x_{i - 1})\mid x_{i - 1} \text{ live}] + \cdot \E[\Delta\mid x_{i - 1} \text{ live}]^2/8 \\
					& \ge \E[\phi(x_{i - 1})\mid x_{i - 1} \text{ live}] + \Omega(\Delta_{j^* / 2^{i-1}}(A)^2). \\
\end{align*}
On the other hand, 
\begin{align*}
\E[\phi(x_i) \mid x_{i - 1} \text{ not live}]& \ge 0 \\
   & = \E[\phi(x_{i - 1}) \mid x_{i - 1} \text{ not live}].
\end{align*}
So we can conclude that
\begin{align*}
\E[\phi(x_i)] & \ge \E[\phi(x_{i - 1})] + \Pr[x_{i - 1} \text{ live}] \cdot \Omega(\Delta_{j^* / 2^i}(A)^2) \\
              & \ge \E[\phi(x_{i - 1})] + 0.9\cdot \Omega(\Delta_{j^* / 2^{i-1}}(A)^2),
\end{align*}
hence \eqref{eq:pd}.
\end{proof}

We can now prove Proposition \ref{prop:imbalances}.
\begin{proof}[Proof of Proposition \ref{prop:imbalances}]
Let $x$ be a random $1$-block. Then by Lemma \ref{lem:phia},
$$\E[\phi(x)] = \Omega\left(\sum_{j \in J} \Delta_j^2(A)\right).$$
On the other hand, $\phi(x) = O(\epsilon^2)$ deterministically. Combined, these imply 
$$\sum_{j \in J} \Delta_j(A)^2 = O(\epsilon^2).$$

\end{proof}

\begin{proof}[Proof of Corollary \ref{cor:imbalances}]
Cauchy-Schwarz implies
$$\sum_{j \in J} \Delta_j(A) ^2\geq \left(\sum_{j \in J} \Delta_j(A)\right)^2/|J| = \left(\sum_{j \in J} \Delta_j(A)\right)^2 / O(\log n).$$ 
Thus we have
$$\sum_{j \in J} \Delta_j(A) \le O(\sqrt{\log n}) \sqrt{\sum_{j \in J} \Delta_j(A) ^2} \le O\left(\epsilon \sqrt{\log n}\right),$$
where the final inequality uses Proposition \ref{prop:imbalances}. Dividing by $|J| = \Theta(\log n)$, we have
$$\frac{1}{|J|} \sum_{j \in J} \Delta_j(A) = O(\epsilon / \sqrt{\log n}).$$
\end{proof}

\subsection{Proof of Proposition \ref{prop:costlower}}\label{sec:costlower}

In this section, we prove Proposition \ref{prop:costlower}. All of the lemmas in this section assume an array of size $m$ that initially contains $n$ elements, where $n = (1 - \epsilon) m$.

We begin by establishing that, if we consider an element with random rank $ t \in [n / 2]$, and we examine the $j$-block beginning at that element, then there are several basic properties that hold with probability at least $0.9$. 
\begin{lemma}
Let $j \in J$ and consider a random $t \in [n / 2 ]$. Define $x$ to be the $j$-block whose first position contains the current rank-$t$ element. With probability at least $0.9$, the following all hold:
\begin{itemize}
\item $x$ is live;
\item $|B(x)| = 3$;
\item $\Delta(x) < c \Delta_j$,
\end{itemize}
\label{lem:nine}
\end{lemma}
\begin{proof}
It suffices to show that each individual property holds with probability at least $0.97$. Observe that $x$ is chosen at random from one of $n/2 \ge m / 3$ $j$-blocks. It therefore suffices to show that, if we define $x'$ to be a uniformly random $j$-block, then each property holds with probability at least $0.99$ for $x'$. 

Note that $x'$ contains at most $\epsilon j$ free slots in expectation, so by Markov's inequality the probability that $x'$ contains $ \ge c \epsilon j$ free slots is at most  $1/c \le 0.01$.  Thus $x'$ is live with probability at least $0.99$.

We claim that $\E[3 - |B(x')|] = 3\epsilon$. This is because the probability that a given slot is occupied is $1-\epsilon$, so $\E[|B(x')|]=3(1-\epsilon)=3-3\epsilon$. Thus, $\E[3 - |B(x')|] = 3\epsilon\le 3/c$. Thus, by Markov's inequality we have $\Pr[3 - |B(x')| \ge 1] =  3/c \le 0.01$. Thus $|B(x')| = 3$ with probability at least $0.99$. 

Finally, observe that $\E[\Delta(x')] = \Delta_j$, so by Markov's inequality we have $\Pr[\Delta(x') \ge c \Delta_j] \le 1/c \le 0.01$. Thus $\Delta(x') < c \Delta_j$ with probability at least $0.99$. 
\end{proof}

Call an insertion/deletion \defn{critical} to a $j$-block $x$ if:  $x$ is live when the operation is performed; and the operation leads to at least one of the elements in $B(x)$ being rearranged. The next lemma argues that, if we perform enough insertions/deletions inside a random $j$-block, then at least one of them will likely be critical.

\begin{lemma}
Let $j \in J$, let $s \in [j / 6, j / 3]$, and consider a random $t \in [n/4 - s, n / 2 - s]$. Define $x$ to be the $j$-block whose first position contains the element with rank $t$. Suppose we perform $\lfloor c j \Delta_j \rfloor + 1$ insertion/deletion pairs, where each insertion adds a new element with rank $t + s$ and each deletion removes the highest-ranked element (i.e., the element with rank $n + 1$). With probability at least $0.6$, at least one of the insertions/deletions is critical to $x$.
\label{lem:eight}
\end{lemma}
\begin{proof}
We know that, with probability at least $0.6$, the properties in Lemma \ref{lem:nine} hold for $x$ both before the insertions/deletions are performed and after the insertions/deletions are performed. Suppose for contradiction that none of the insertions/deletions are critical to $x$.

Thus, we know that none of the operations rearrange any of the elements in $B(x)$. We additionally claim that none of the elements in $B(x)$ are deleted. This follows from the fact that we always delete the element of rank $n+1$, while the highest-ranked element in $x$ is has rank less than $n+1$ for the following reason. The first element of $x$ is at rank at most $n/2$, and $x$ contains at most $j^*=2^{\lfloor \log m\rfloor-2}\leq \frac{n}{4(1-\epsilon)}$ elements. So the last element of $x$ has rank at most $\frac{n}{2}+\frac{n}{4(1-\epsilon)}$, which is less than $n+1$ since $\epsilon<1/2$.

Additionally, we claim that all of the insertions go into the first $j / 2$ slots of $x$. This is because otherwise there would be at most $s\le j/3$ elements in the first $j/2$ slots, which means that there would be least $j/6$ empty slots, which contradicts the fact that $x$ is live. 

Thus, during the course of the insertions/deletions, the first $j / 2$ slots in $x$ gain $\lfloor  c j \Delta_j \rfloor + 1$ elements, while the second $j / 2$ slots of $x$ remain stable in their number of elements. We know that $\Delta(x) < c\Delta_j$ both before and after the insertions/deletions are performed. So, over the course of the insertions/deletions, the density of the first $j/2$ slots in $x$ changes by less than $2c\Delta_j$. This means that the number of elements in the first $j/2$ slots of $x$ changes by at most $j/2 \cdot 2c\Delta_j=cj\Delta_j$, which contradictions the fact that the first $j / 2$ slots in $x$ gain $\lfloor  c j \Delta_j \rfloor + 1$ elements.
\end{proof}

Using symmetry, we can reinterpret the previous lemma as a statement about a single insertion/deletion pair.
\begin{lemma}
Let $j \in J$, let $s \in [j / 6, j / 3]$, and consider a random $t \in [n/4 - s, n / 2 - s]$. Define $x$ to be the $j$-block whose first position contains the element with rank $t$. Suppose we perform a single insertion/deletion pair, where the insertion adds a new element with rank $t + s$ and the deletion removes the current highest-ranked element (i.e., the element with rank $n + 1$). With probability $\Omega\left(\frac{1}{\lfloor j \Delta_j \rfloor + 1}\right)$, at least one of the insertion/deletion is critical to $x$. 
\label{lem:probcritical}
\end{lemma}
\begin{proof}
By history independence, the probability distribution of array configurations is only dependent upon $n$, $m$, and $r$, and these quantities are the same after each insertion/deletion pair in Lemma \ref{lem:eight}. Thus, Lemma \ref{lem:eight} immediately extends to each individual insertion/deletion pair.
\end{proof}

The previous lemma analyzes for a specific block $x$ the probability that a specific insertion/deletion pair is critical to $x$. Notice, however, that a given insertion/deletion pair can be critical to many $j$-blocks simultaneously. Indeed, by applying Lemma \ref{lem:eight} simultaneously for multiple different values of $s$, we can deduce a lower bound on the expected number of elements that are rearranged at distance $\Theta(j)$ (in rank) from the element currently being inserted.
\begin{lemma}
Let $j \in J$. Suppose we perform an insertion at a random rank $r \in [n / 4, n / 2]$ and then we delete the highest-ranked element (i.e., the element with rank $n + 1$). Let $q_j$ be the number of elements that are rearranged by the insertion/deletion, and that have ranks $r'$ satisfying $|r - r'| = \Theta(j)$ after the insertion. Then 
$$\E[q_j] = \Omega\left(\frac{1}{\Delta_j + 1/j}\right).$$
\label{lem:jdist}
\end{lemma}
\begin{proof}
For $s \in [j / 6, j / 3]$, define $x_s$ to be the $j$-block beginning with the element whose rank is $r - s$. Note that the sets $B(x_s)$ are disjoint across  $s \in [j / 6, j / 3]$. Let $B_s$ be the number of elements in $B(x_s)$ that are rearranged by the insertion/deletion; and let $E_s$ be the event that both $B_s \ge 1$ and that $x_s$ is live. 

If $x_s$ is live then the elements of $B_s$ have ranks $r'$ satisfying $|r - r'| = \Theta(j)$. Since the $B_s$'s are disjoint, it follows that
$$\E[q_j] \ge \sum_{s \in [j / 6, j / 3]} \Pr[E_s].$$

For each $s \in [j / 6, j / 3]$, we have by Lemma \ref{lem:probcritical} that
$$\Pr[E_s] = \Omega\left(\frac{1}{\lfloor j \Delta_j \rfloor + 1}\right).$$
Thus 
$$\E[q_j] = \Omega\left(\frac{j}{\lfloor j \Delta_j \rfloor + 1}\right) = \Omega\left(\frac{1}{\Delta_j + 1/j}\right),$$ as desired.
\end{proof}

Finally, we can deduce a lower bound on the total number of elements that are rearranged by a random insertion/deletion.
\begin{lemma}
Suppose we perform an insertion at a random rank $r \in [n / 4, n / 2]$ and then we delete the highest-ranked element. The expected total cost of the insertion/deletion is at least
$$\Omega\left(\sum_{j \in J} \frac{1}{\Delta_j + 1/j}\right).$$
\label{lem:costlower}
\end{lemma}

\begin{proof}
Let $q$ be the number of elements that are rearranged by the insertion/deletion.  For each $j \in J$, define $q_j$ as in Lemma \ref{lem:jdist}.
Each element that is rearranged by the insertion/deletion has a rank $r'$ satisfying $|r' - r| = \Theta(j)$ for at most a constant number of $j \in J$. That is, each rearrangement is counted by at most $O(1)$ of the $q_j$'s. Thus  
$$q = \Omega\left(\sum_j q_j\right).$$
By Lemma \ref{lem:jdist}, it follows that
$$\E[q] = \Omega\left(\sum_{j \in J} \frac{1}{\Delta_j + 1/j}\right).$$
\end{proof}

 Lemma \ref{lem:costlower} considers an insertion with a random rank $r \in [n/4, n / 2]$, but this trivially implies the same claim for a random rank $r \in \{0, \ldots, n\}$ (i.e., Proposition \ref{prop:costlower}). Thus the section is complete. \section{Upper Bound For Sparse Arrays}
\label{sec:sparse}

Define the \defn{$\tau$-sparse list-labeling problem} to be the list-labeling problem in the regime of $n \le m / \tau$. Previously in this paper, we studied the setting where $\tau = O(1)$. In this section, we extend our upper bounds to apply to the sparse regime where $m = \tau n$ for some $16 \le \tau \le n^{o(1)}$. We do this via a simple general-purpose reduction from the sparse setting to the linear setting.

We will prove the following proposition:

\begin{proposition}

Let $T$ be a non-negative convex function  satisfying $T(\Theta(i)) = \Theta(T(i))$ for all $i$ and satisfying $T(0) = 0$. Let $16 \le \tau \le n^{o(1)}$. If there exists a 2-sparse list-labeling solution whose expected amortized cost is upper bounded by $T(\log n)$, then there exists a $\tau$-sparse list-labeling solution whose expected amortized cost is upper bounded by
$$O\left(T\left(\frac{\log n}{\log \tau}\right) \cdot \log \tau\right).$$
\label{prop:sparsered}
\end{proposition}

Combining Proposition \ref{prop:sparsered} and Corollary \ref{thm:linearcase}, we obtain the following upper bound for the sparse regime:
\begin{theorem}
For $16 \le \tau \le n^{o(1)}$, there exists a solution to the $\tau$-sparse list-labeling problem with expected amortized cost upper bounded by 
$$O\left(\frac{\log^{3/2} n}{\sqrt{\log \tau}}\right).$$
\label{thm:sparse}
\end{theorem}

To prove Proposition \ref{prop:sparsered}, we introduce an intermediate problem that we call the \defn{bucketed list-labeling problem}. In this problem, there are $m$ buckets and up to $N = \Omega(m)$ elements at a time, with elements being inserted and deleted as in the classical list-labeling problem. Elements must be assigned to buckets so that, if two elements $a$ and $b$ are assigned to buckets $u \neq v$, then $a < b \iff u < v$. The cost of adding/removing an element to/from a bucket is $0$ when that element is inserted/deleted, but the cost of rearranging items is equal to the \emph{sum} of the sizes of the buckets containing those items. (So even moving one item from a bucket $u$ to a bucket $v$ costs $|u| + |v|$). Finally, each bucket has a maximum capacity of $8 N / m$ elements.

Our next lemma reduces bucketed list labeling to 2-sparse list labeling.
\begin{lemma}
Let $T$ be a non-negative convex function.  If there exists a 2-sparse list-labeling solution whose expected amortized cost is upper bounded by $T(\log n)$, then there exists a bucketed list-labeling solution whose expected amortized cost is upper bounded by
$$O\left(T\left(\log m\right)\right).$$
\label{lem:rone}
\end{lemma}
\begin{proof}
An important component of our bucketed list-labeling solution is to partition the elements into up to $m / 2$ disjoint \defn{blocks}, where each block contains up to $8N / m$ consecutive elements. We maintain these blocks using hysteresis: every time that a block's size falls below $2 N / m$ (due to deletions), we merge it with an adjacent block (unless there is only one block in the system); and every time that a block's size exceeds $8N / m$ (due to insertions or merges), we split that block into two blocks of equal size. Note that a block's size can never exceed $10N / m$ because a block of size $\leq 8N / m$ can be merged with a block of size $<2 N / m$, and there is no way to create a larger block. Thus, after a split, the size of each resulting block is between $4N / m$ and $5N / m$. Starting from an empty array, during a sequence of $k$ insertions/deletions, the number of block splits/merges will be at most $O(k m / N)$.

To construct a bucketed list-labeling solution, we treat the $m$ buckets as slots in an array of size $m$, and we treat the up-to-$m / 2$ blocks as elements that reside in that array. This allows for us to treat the bucketed list-labeling problem as a 2-sparse list-labeling problem: block splits corresponded to element insertions in the 2-sparse list-labeling problem; and block merges correspond to element deletions in the 2-sparse list-labeling problem. 

If an operation incurs cost $S$ in the 2-sparse list-labeling problem, then it incurs cost $O(S \cdot N / m)$ in the bucketed list-labeling problem (since each element in the former problem corresponds to a block of $O(N / m)$ elements in the latter problem). On the other hand, starting from an empty array, if $k$ insertions/deletions are performed in the bucketed list-labeling problem, the number of insertions/deletions in the 2-sparse list-labeling problem will only be $O(k m / N)$.  Combining these with the assumption that the 2-sparse list-labeling problem incurs cost $T(\log n)$, we have that the total cost of the bucketed list-labeling problem is $O(T(\log(m/2))\cdot N/m \cdot k \cdot m/N) = O(kT(\log m))$, thus the amortized cost of the bucketed list-labeling problem is $O(T(\log m))$.
\end{proof}

Next we reduce sparse list labeling to bucketed list labeling.
\begin{lemma}
Let $T$ be a non-negative convex function satisfying $T(\Theta(i)) = \Theta(T(i))$ for all $i$ and satisfying $T(0) = 0$. Let $16 \le \tau \le n^{o(1)}$. If there exists a bucketed list-labeling solution whose expected amortized cost is upper bounded by $T(\log m)$, then there exists a $\tau$-sparse list-labeling solution whose expected amortized cost is upper bounded by
$$O\left(T\left(\frac{\log n}{\log \tau}\right) \cdot \log \tau\right).$$
\label{lem:rtwo}
\end{lemma}
\begin{proof}
We may assume without loss of generality that $\tau$ is a natural number. We prove the result by induction on $\tau$. The base case of $16 \le \tau \le O(1)$ is trivial, since we can break the array into $\Theta(n)$ chunks of size $\Theta(1)$ and treat each chunk as a bucket in the bucketed list-labeling problem.

Now suppose that $\omega(1) \le \tau \le n^{o(1)}$. Let $c$ be a large positive constant (to be selected later), and partition the array into $n^{c/\log \tau}$ chunks of size $m ' = \lfloor m / n^{c/\log \tau}\rfloor$ slots each (possibly orphaning $O(n^{c / \log \tau})$ slots due to rounding errors). Treat each of the $n^{c/\log \tau}$ chunks as a bucket, and assign elements to chunks using bucketed list labeling. By assumption, bucketed list labeling has expected amortized cost $T(\log m)$ where $m$ is the number of buckets, and plugging in the value $n^{c/\log \tau}$ for $m$ we get that the expected amortized cost of the bucketed list labeling instance is: 
$$T(\log n^{c/\log \tau}) = T\left( \frac{c \log n}{\log \tau}\right)$$
per operation. Since $T(\Theta(i)) = \Theta(T(i))$, we can further bound the above cost to be at most 
\begin{equation}
    c' \cdot T\left( \frac{\log n}{\log \tau}\right),
    \label{eq:sparsea}
\end{equation}
where $c'$ is a constant determined by $c$.

By design, each chunk contains at most $n' = 8 n / n^{c/\log \tau}$ elements, so 
\begin{align*}
    \frac{m'}{n'} \ge \frac{\lfloor m / n^{c/\log \tau}\rfloor}{8 n / n^{c/\log \tau}} \ge \frac{m}{16n} \ge \tau / 16.
\end{align*}
Thus we can recursively implement each chunk as an instance of $\frac{\tau}{16}$-sparse list labeling. By our inductive hypothesis for $\tau' = \frac{\tau}{16}$, we have that for every sufficiently large positive constant $Q$, the expected amortized cost of performing an insertion/deletion in a given chunk is at most 
\begin{align*}
& Q \cdot T\left(\frac{\log n'}{\log \tau'}\right) \cdot \log \tau' \\
& = Q \cdot T\left(\frac{(1 - c / \log \tau + 3/\log n) \log n }{(1 - 4/\log \tau)\log \tau}\right) \cdot (1-4/\log \tau) \log \tau \\
& \le Q \cdot T\left(\frac{(1 - 1 / \log \tau) \log n }{\log \tau}\right) \cdot \log \tau \hfill \hspace{7mm}\text{(since }c\ge 8\text{)}\\
& \le Q \cdot (1 - 1 / \log \tau) \cdot T\left(\frac{\log n }{\log \tau}\right) \cdot \log \tau \hspace{5mm}\text{(since }T\text{ is convex and }T(0) = 0\text{)}\\
& \le Q \cdot T\left(\frac{\log n }{\log \tau}\right) \cdot \log \tau - Q \cdot  T\left(\frac{\log n }{\log \tau}\right).
\end{align*}
Combining this with \eqref{eq:sparsea}, the total expected amortized cost of an insertion/deletion is at most 
\begin{align*}
& c' \cdot T\left( \frac{\log n}{\log \tau}\right) +  Q \cdot T\left(\frac{\log n }{\log \tau}\right) \cdot \log \tau - Q \cdot  T\left(\frac{\log n }{\log \tau}\right). \\
\end{align*}
Choosing $Q$ to be at least $c'$, this is at most
$$Q \cdot T\left(\frac{\log n }{\log \tau}\right) \cdot \log \tau,$$
which completes the proof by induction.
\end{proof}

Lemmas \ref{lem:rone} and \ref{lem:rtwo} directly imply Proposition \ref{prop:sparsered}, completing the section. \section{Related work}\label{sec:related}

\paragraph{Formulations and reformulations.}
The list-labeling problem has been independently formulated several times and under various names. It was first studied by Itai, Konheim and Rodeh~\cite{ItaiKoRo81} as a sparse table scheme for implementing priority queues. Willard~\cite{Willard81} considered the \emph{file-maintenance problem}, where records are inserted and deleted in a sequentially ordered file. Dietz~\cite{Dietz82} formulated the similar  \emph{order-maintenance problem} of maintaining order in a linked list with efficient insertions. Andersson~\cite{Andersson89} and Andersson and Lai~\cite{AnderssonLa90} studied a version of the problem in the context of balanced binary search trees, which Galperin and Rivest~\cite{GalperinR93} independently studied under the name \emph{scapegoat trees}. Raman~\cite{Raman99} posited an analogous problem related to building locality preserving dictionaries. 

This problem has mainly been studied in four regimes for the size $m$ of the label array: dense ($m=(1+o(1))n$), linear ($m=(1 + \Theta(1))n$), polynomial ($m=n^{1 + \Theta(1)}$), and superpolynomial ($m=n^{\omega(1)}$).

\paragraph{Upper and lower bounds in the linear regime.}
In the linear regime, Itai, Konheim and Rodeh\cite{ItaiKoRo81}, first proved that items can be inserted with $O(\log^2 n)$ amortized cost. Various subsequent works have made improvements or simplifications to the algorithms achieving this cost, but the upper bound has remained unchanged. Willard~\cite{Willard82, Willard86, Willard92} deamortized this result to a $O(\log^2 n)$ worst-case cost. Bender, Cole, Demaine, Farach-Colton and Zito~\cite{BenderCoDe02twosimplified},  Bender, Fineman, Gilbert, Kopelowitz and Montes~\cite{BenderFiGi17} and Katriel~\cite{Katriel02} provided simplified algorithms for this result for the {order-maintenance problem}. Itai and Katriel~\cite{ItaiKa07} additionally simplified the algorithm for the amortized upper bound. 

The list-labeling problem where $m=(1+\varepsilon)n$, and where the gap between any two inserted items is $O(1)$ is often called the \emph{packed-memory array problem}, for which bounds of $O(\varepsilon^{-1} \log^2 n)$ are known~\cite{BenderDeFa05,BenderDeFa00,BenderFiGi05}.  
Bender and Hu~\cite{BenderHu07} provided an \emph{adaptive} packed-memory array algorithm, that is, it matches the $O(\log^2n)$  worst case insertion cost in the linear regime while achieving cost of $O(\log n)$ on certain common classes of instances. 
Bender, Berry, Johnson, Kroeger, McCauley, Phillips, Simon, Singh and Zage~\cite{BenderBeJo16} presented a history-independent packed-memory array which again matches the existing upper bound in the linear regime.

Dietz and Zhang~\cite{dietz1990lower} proved a lower bound on insertion costs of $\Omega(\log^2 n)$ amortized per insertion in the linear regime for the natural class of \emph{smooth} algorithms, where the relabelings are restricted to evenly rebalance elements across a contiguous subarray.  Bul\'anek, Kouck\'y and Saks.~\cite{BulanekKoSa12} showed a $\Omega(\log^2 n)$ lower bound for deterministic algorithms in the linear regime, and thus proved that the best known upper bounds were tight for deterministic algorithms.  The best general lower bound is $\Omega(\log n)$ in the linear regime~\cite{BulanekKoSa13}.

\paragraph{Other upper bounds.}
In the dense setting, Andersson and Lai~\cite{AnderssonLa90}, Zhang~\cite{zhang1993density}, and Bird and Sadnicki~\cite{BirdSa07} showed an $O(n\log^3n)$ upper bound for filling an array from empty to full for $m=n$. For arrays of polynomial size, it was known as a folklore algorithm that an amortized $O(\log n)$ insertion cost can be achieved by modifying the techniques in~\cite{ItaiKoRo81}. Kopelowitz~\cite{Kopelowitz12} extended this to a worst case upper bound. This bound was also matched in the balanced search tree setting~\cite{Andersson89,GalperinR93}. In the superpolynomial array regime, Babka, Bul\'anek, Cun\'at, Kouck\'y and Saks~\cite{BabkaBCKS19} showed an algorithm with amortized $O(\log n/\loglog m)$ cost when $m=\Omega(2^{\log^k n})$, which implies constant amortized cost in the pseudo-exponential regime of $m = 2^{n^{\Omega(1)}}$. 
Devanny, Fineman, Goodrich and Kopelowitz~\cite{DevannyFiGo17} studied the \emph{online house numbering problem}, which is similar to the list-labeling problem, except with the objective to minimize the maximum number of times an element is relabeled.

\paragraph{Other lower bounds.} 
Dietz and Zhang~\cite{dietz1990lower} proved a lower bound of $\Omega(\log n)$ per insertion in the polynomial regime for \emph{smooth} algorithms. Bul\'anek, Kouck\'y and Saks~\cite{BulanekKoSa12} showed an $\Omega(n\log^3n)$ lower bound for $n$ insertions into an initially empty array of size $m=n+n^{1-\epsilon}$. Dietz, Seiferas and Zhang\cite{dietz2004tight} proved a lower bound of $\Omega(\log n)$ in the polynomial regime for general deterministic algorithms, with a simplification by Babka, Bul\'anek, Cun\'at, Kouck\'y, and Saks~\cite{BabkaBCKS12}.
Bulanek, Kouck\'y and Saks~\cite{BulanekKoSa13} also proved that the $\Omega(\log n)$ lower bound for the polynomial regime extends to randomized algorithms. In the superpolynomial regime, Babka, Bul\'anek, Cun\'at, Kouck\'y and Saks~\cite{BabkaBCKS19} showed a lower bound of $\Omega\left(\frac{\log n}{\loglog m - \loglog n}\right)$ for $m$ from $n^{1+C}$ to $2^n$, which reduces to a bound of $\Omega(\log n)$ for $m=n^{1+C}$. 

\paragraph{Theoretical Applications.}
Applications of list labeling include the diverse motivating problems under which it was first studied, such as priority queue implementation, ordered file maintenance, etc. 
Hofri and Konheim~\cite{HofriK87} studied a similar array structure for use in a \emph{control density array}, a sparse table that supports search, insert and deletion by keys. Fagerberg, Hammer and Meyer~\cite{FagerbergH019} used upper bounds from~\cite{ItaiKoRo81} for their rebalancing scheme, which maintains optimal height in a balanced B-tree. 

Bender, Demaine and Farach-Colton~\cite{BenderDeFa05} used the packed-memory array in their \emph{cache-oblivious B-tree} algorithm, so our result directly implies an improvement in that scheme.  Specifically, insertions into their B-tree take $O(\log_B N + (\log^2 N)/B)$ I/Os, and using our list-labeling algorithm, this is improved to  $O(\log_B N + (\log^{3/2} N)/B)$ I/Os. Brodal, Fagerberg and Jacob~\cite{BrodalFaJa02} and Bender, Duan, Iacono and Wu~\cite{BenderDuIa04} independently simplified the cache-oblivious B-tree algorithm. Bender, Fineman, Gilbert and Kuszmaul~\cite{BenderFiGi05} presented \emph{concurrent} cache-oblivious B-trees for the distributed setting. Bender, Farach-Colton and Kuszmaul~\cite{BenderFaKu06} described \emph{cache-oblivious string B-trees} for improved performance on variable length keys, compressed keys, and range queries.   All of these cache-oblivious algorithms use packed-memory arrays.

In their results on the \emph{controller problem} for managing global resource consumption in a distributed network, Emek and Korman~\cite{EmekKo11} reduced the list-labeling problem to prove their lower bounds. 
Bender, Cole, Demaine, Farach-Colton and Zito~\cite{BenderCoDe02twosimplified} also applied list labeling lower bounds to the problem of maintaining a dynamic ordered set which supports traversals in the cache-oblivious and sequential-access models. 
Kopelowitz~\cite{Kopelowitz12} studied the \emph{predecessor search on dynamic subsets of an ordered dynamic list problem}, which combines the order-maintenance problem with the \emph{predecessor problem} of maintaining dynamic sets which support predecessor queries.
Nekrich used techniques for linear list labeling  from~\cite{ItaiKoRo81} in data structures supporting various problems related to querying points in planar space, such as orthogonal range reporting~\cite{Nekrich07, Nekrich09}, the stabbing-max problem\cite{Nekrich11}, and the related problem of searching a dynamic catalog on a tree~\cite{Nekrich10}. Mortensen~\cite{Mortensen03} similarly considered applications to the orthogonal range and dynamic line segment intersection reporting problems.

\paragraph{Practical Applications.}
Additionally, a variety of practical applications use the packed-memory array as an algorithmic component.  Durand, Raffin and Faure~\cite{DurandRF12} proposed using a packed-memory array to maintain sorted order during particle movement simulations for efficient searching. Khayyat, Lucia, Singh, Ouzzani, Papotti, Quian\'e-Ruiz, Tang and Kalnis~\cite{KLSOPQ0K17} applied it to handle dynamic database updates in their inequality join algorithms. Toss, Pahins, Raffin and Comba~\cite{TossPRC18} presented a \emph{packed-memory quadtree}, which supports large streaming spatiotemporal datasets. De Leo and Boncz~\cite{LeoB19pma} presented the \emph{rewired memory array}, an implementation of a packed-memory array which improves on its practical performance. Several works \cite{WheatmanX21, WheatmanX18, WheatmanB21, PandeyWXB21, LeoB21, LeoB19fastconcurrent} implemented \emph{parallel} packed-memory arrays for the purpose of storing dynamic graphs with fast updates and range queries. Assessing whether our results can be used to obtain practical speedups for these applications remains an interesting direction for future work.

\paragraph{Related work on history independence.}
History independence has been studied for data structures in both internal and external memory models~\cite{Micciancio97, NaorTe01, HartlineHoMo05, BuchbinderPe03, BlellochGo07, NaorSeWi08, Golovin09, Golovin10, BenderBeJo16}. Even prior to the formalization of history independence~\cite{Micciancio97, NaorTe01} in the late 1990s, there were several notable early works on hashing and search trees that implicitly achieved history-independent topologies~\cite{AmbleKn74,SundarTa90,AnderssonOt91,Snyder77,Pugh88,Pugh90,AragonSe89,PughTe89}. 
The notion of history independence studied in this paper is sometimes referred to as \emph{weak history independence}---for a survey of stronger notions of history independence, along with other related work, see recent work~\cite{GoodrichKoMi17} by Goodrich, Kornaropoulos, Mitzenmacher and Tamassia. (Note that, the \emph{weaker} the notion of history independence that one uses, the \emph{stronger} any lower bound on history-independent data structures becomes.)

History independence is typically treated as a security property: the goal is to minimize the amount of information that is leaked if an adversary sees internals of the data structure.  To the best of our knowledge, the results in this paper are the first to use techniques from history independence in order to achieve \emph{faster} algorithms than were previously possible.

\section{Acknowledgements}

This research was partially sponsored by the United States Air Force Research Laboratory and the United States Air Force Artificial Intelligence Accelerator and was accomplished under Cooperative Agreement Number FA8750-19-2-1000. The views and conclusions contained in this document are those of the authors and should not be interpreted as representing the official policies, either expressed or implied, of the United States Air Force or the U.S. Government. The U.S. Government is authorized to reproduce and distribute reprints for Government purposes notwithstanding any copyright notation herein.

This work was also supported by NSF grants 
CCF-2106999, 
CCF-2118620, 
CNS-1938180, 
CCF-2118832,  
CCF-2106827, 
CSR-1763680, 
CCF-1716252, 
CNS-1938709. Nicole Wein is supported by a grant to DIMACS from the Simons Foundation (820931).  
Finally, William Kuszmaul is partially supported by a Hertz Fellowship and an NSF GRFP Fellowship. 

\bibliographystyle{plain}

\appendix

\section{Proof of Lemma \ref{lem:overflow}}\label{app:over}

\overflow*

We remark that Lemma \ref{lem:overflow} is essentially immediate from the fact that each subproblem $S$ has density $\mu_S \le 1 + \epsilon / 2$. The only difficulty in the proof comes from the necessity to carefully handle floors/ceilings.

\begin{proof}
By construction, each level-$i$ subproblem $S$ has $$|C_S| \le 2 \alpha_S  \cdot \frac{n} {2^i} \le \epsilon \cdot \frac{n}{2^i}.$$ Thus, if $|C_S| > 4$ (i.e., $S$ is a non-base-case subproblem), we must have $\epsilon n / 2^i \ge 4$. Since every base-case subproblem is the child of a non-base-case subproblem, we have that for base-case subproblems $\epsilon n/2^{i-1} \ge 4$. This means that every subproblem $S$ is in a level $i$ satisfying 
\begin{equation}
\frac{\epsilon n }{2^i} \ge 2.
\label{eq:sizereq}
\end{equation}
We wish to show that $\lfloor \frac{m}{2^i} \rfloor - |S| \ge 0$. We know that
$$\left\lfloor \frac{m}{2^i} \right\rfloor - |S|  = \left\lfloor \frac{m}{2^i} \right\rfloor - \mu_S \frac{n}{2^i}  \ge \frac{m}{2^i} - \mu_S \frac{n}{2^i} - 1  \ge \frac{(1 + \epsilon)n - \mu_S n}{2^i} - 1.$$
Since $\mu_S \le 1 + \epsilon / 2$, it follows that
$$\left\lfloor \frac{m}{2^i} \right\rfloor - |S|  \ge \frac{\epsilon n / 2}{ 2^i}  - 1.$$
By \eqref{eq:sizereq}, we can conclude that $\lfloor \frac{m}{2^i} \rfloor - |S| \ge 0$, as desired. 
\end{proof}
 
\end{document}